\documentclass[twoside, 11pt]{article}
\pdfoutput=1

\usepackage[preprint]{jmlr2e}
 
\usepackage{lastpage}
\jmlrheading{23}{2022}{1-\pageref{LastPage}}{1/21; Revised 5/22}{9/22}{21-0000}{Victor S. Portella and Nick Harvey}
\editor{My editor}

\ShortHeadings{Lower Bounds for Private Estimation of Gaussian Covariance Matrices}{Victor S. Portella and Nick Harvey}
\firstpageno{1}

\usepackage{times}

\usepackage{packages}

\newcommand{\WellCond}{\mathcal{W}}
\newcommand{\PSDCone}{\mathbb{S}_+^d}
\newcommand{\PDCone}{\mathbb{S}_{++}^d}
\newcommand{\bd}{\mathrm{bd}}

\newcommand{\repeatclaim}[2]{\medskip\noindent\textbf{#1. }{#2} 
}

\newcommand{\PrSigma}{\operatorname{Pr}_{|\Sigma}}
\newcommand{\ESigma}[1]{\operatorname{E}_{|\Sigma}\left[\,#1\,\right]}  
\newcommand{\ESigmaSmall}[1]{\operatorname{E}_{|\Sigma}[\,#1\,]} 
\newcommand{\ESigmaBig}[1]{\operatorname{E}_{|\Sigma}\Big[\,#1\,\Big]} 

\newenvironment{sketch}%
    {\begin{proof}[Sketch]}
    {\end{proof}}
\newcommand{\lmin}{\lambda_{\min}}

\numberwithin{theorem}{section} 

\title{
    Lower Bounds for Private Estimation of Gaussian Covariance Matrices under All Reasonable Parameter Regimes
    }

\author{\name Victor S.\ Portella \email victorsp@cs.ubc.ca \\
\name Nicholas J.\ A.\ Harvey \email nickhar@cs.ubc.ca \\
\addr University of British Columbia%
}

\begin{document}

\maketitle

\vspace{-1em}

\begin{abstract}%
  We prove lower bounds on the number of samples needed to privately estimate the covariance matrix of a Gaussian distribution.
  Our bounds match existing upper bounds in the widest known setting of parameters. Our analysis relies on the Stein-Haff identity, an extension of the classical Stein's identity used in previous fingerprinting lemma arguments.  
\end{abstract}

\begin{keywords}%
  Differential privacy, fingerprinting, Stein Identity, lower bounds, covariance estimation%
\end{keywords}

\section{Introduction}

Differential Privacy (DP) is a widely adopted framework to perform data analysis while
avoiding
leakage of sensitive information \citep{DworkMNS06a}.
A major thrust of research in DP is 
developing privacy preserving algorithms for a variety of fundamental problems in computer science and statistics.
In the past few years, a direction of particular interest has been parameter estimation of probability distributions.

Multivariate Gaussians are perhaps the canonical distribution for which to study parameter estimation. Indeed, there has been considerable work on algorithms to estimate the mean and covariance of a Gaussian distribution under both pure and approximate differential privacy (see~\Section{RelatedWork} or \citealp[\S 1.2]{Hopkins0MN23a} for many examples). To understand whether these algorithms are optimal, we require lower bounds for the error under DP. For mean estimation with a known covariance matrix there are lower bounds that match existing sample complexity bounds \citep{Kamath0SU19a,Aden-AliA021}. However, the case of covariance matrix estimation under \((\eps, \delta)\)-DP is, at present, not completely understood.

The current best $(\eps,\delta)$-DP sample complexity bound for estimating a $d \times d$ covariance matrix
up to \(\alpha\) error in Frobenius norm
is \(n = \tildeO(d^2/\alpha^2 + d^2/\alpha \eps + \log(1/\delta)/\eps)\) samples, due to
\citet{Aden-AliA021}.
(For convenience throughout this section we restrict attention to covariance matrices with all eigenvalues in $\Theta(1)$.)
Regarding lower bounds, $\Omega(d^2/\alpha^2)$ samples are needed even without privacy, and at least \(\Omega(\log(1/\delta)/\eps)\) samples are necessary \citep{KarwaV18} with \((\eps, \delta)\)-DP.
 Recently,~\citet{KamathMS22a} and \citet{Narayanan23a} have shown \(n = \tilde{\Omega}(d^2/\alpha \eps)\) lower bounds for some regimes of \(d,n\) and \(\delta\).
Our work also proves the lower bound \(n = \Omega(d^2/\alpha \eps)\), but in the most general setting of parameters.
The currently best lower bounds may be separated into two regimes.
\begin{itemize}
    \item \emph{High-accuracy regime:} the error in Frobenius norm is \(\alpha = O(1)\). In this regime, \citet{KamathMS22a} shows that if \(\delta = \tildeO(1/n)\), then \(n = \Omega(d^2/\eps \alpha)\);
    \vspace{-0.5em}
    \item \emph{Low-accuracy regime:} the error in Frobenius norm is \(\alpha = O(\sqrt{d})\). (For larger $\alpha$ the problem is trivial.) In this regime \citet{Narayanan23a} shows that, if \(\delta = O(1/d^2)\), then \(n = \tilde{\Omega}(d^2/\eps \alpha)\). This result has less restrictive hypothesis on $\alpha$, but more restrictive hypothesis on $\delta$ since it is quadratic.
\end{itemize} 
In summary, if \(\delta\) is as large as \(\tildeO(1/n)\) and \(n = \Omega(d^2)\), then the sample complexity of~\citet{Aden-AliA021} guarantees \(O(1)\) error and is optimal.
On the other hand, for \(n = o(d^2)\), the mechanism of \citet{Aden-AliA021} may have accuracy $\alpha=\Omega(1)$. In this accuracy regime, the known lower bound holds only if \(\delta = O(1/d^2)\) \citep{Narayanan23a}.
Our work completes the picture, showing \(n = \Omega(d^2/\alpha\eps)\), without any logarithmic factors, without requiring any regularity conditions on the mechanism, under both accuracy regimes, and with \(\delta = O(1/n \ln n)\).

\subsection{Our contributions}

Our lower bounds are based on the mechanism's error \(\alpha\)  when the random\footnote{Although mechanisms for this problem are designed to work with inputs drawn from a distribution with a \emph{fixed} covariance matrix $\Sigma$, we will choose $\Sigma$ randomly so that it is unknown to the mechanism. Randomizing the parameters the mechanism aims to estimate is usually a fundamental part of fingerprinting arguments. The distribution on $\Sigma$ is specified in \Section{ProofOverview}.} 
covariance matrix $\Sigma$ is ``well-conditioned'', by which we mean that all eigenvalues are in $[0.09,10]$ (see~\eqref{eq:alphadef} for a formal definition of \(\alpha\)); the mechanism's performance may be arbitrarily poor otherwise.

\begin{theorem}[Main Theorem]
\TheoremName{Main}
    Let \(\cM \colon (\bR^d)^n \to \bR^{d \times d}\) be \((\eps, \delta)\)-DP with  \(\eps \in (0,1)\) and 
    \begin{equation}
        \EquationName{DeltaAssumption}
        \delta \leq \frac{1}{3 n \ln(en)}.
        \end{equation}
    Let \(\Sigma\) be a random positive definite matrix in \(\bR^{d \times d}\) and let \(\alpha^2\) be the squared error of \(\cM\) given by expected squared Frobenius norm error conditioned on the event that $\Sigma$ has eigenvalues in \([0.09, 10]\). Then, there is a distribution of \(\Sigma\) such that, if $\alpha \in [2^{-d},\frac{\sqrt{d}}{15}]$, then
    \[
    n ~=~ \Omega\Big(\frac{d^2}{\eps \alpha}\Big).
    \]
\end{theorem}

Our analysis extends the score attack~\citep{CaiWZ23a} to a distribution with non-independent parameters, and leads to a quantitatively stronger result than previous work~\citep{KamathMS22a,Narayanan23a} with a mathematically clean argument involving a generalization of the Stein identity, known as the Stein-Haff identity~\citep{Haff79a}. 
In addition, our bounds have no polylog factors and do not require any regularity conditions on the mechanism. We believe these techniques suggest a general strategy to use fingerprinting arguments without relying on independence.


\subsection{Related Work} 
\SectionName{RelatedWork}

\paragraph{Covariance Matrix Estimation.} To learn a Gaussian it is sufficient~\citep[\S 2.1]{AshtianiL22} and (in some sense) necessary~\citep[Thm.~1.8]{ArbasAL23} to estimate the mean and covariance matrix. For learning Gaussians under approximate differential privacy, \citet{KarwaV18} showed polynomial-time algorithms to learn unbounded \emph{one dimensional} Gaussians. Following their work, there were a series of works on Gaussian covariance estimation under approximate DP~\citep{Aden-AliA021,TsfadiaCKMS22, LiuKO22,   AshtianiL22, KamathMSSU22, Hopkins0MN23a}, concentrated DP~\citep{Kamath0SU19a},  and pure DP~\citep{BunKSW21,Hopkins0MN23a}.
See~\citet[Table~1]{Hopkins0MN23a} for a summary of the results on Gaussian covariance estimation. The best sample complexity known to approximate a Gaussian covariance estimation up to \(\alpha\) error in Mahalanobis norm under \((\eps, \delta)\)-DP is \(n = \tildeO(d^2/\alpha^2 + d^2/\alpha \eps + \log(1/\delta)/\eps)\) due to~\citet{Aden-AliA021}, with a polynomial time algorithm recently proposed by~\citet{Hopkins0MN23a}. \citet[Thm.~1.4]{KarwaV18} shows an \(\Omega(\log(1/\delta)/\eps)\) lower bound for learning one dimensional Gaussians, and \citet[Theorem 56]{Kamath0SU19a} shows an \(\Omega(d/\alpha \eps)\) lower bound for learning spherical Gaussians. Thus, it only remains to show a \(\Omega(d^2/\alpha \eps)\) lower bound to conclude that the currently best-known sample complexities are not improvable (up to poly-logarithmic factors). A related problem is estimating the empirical covariance matrix out of worst-case (bounded) data (e.g.,~\citealp{DworkTTZ14a, DongLY22a}, see \citealt[\S 6]{Narayanan23a} for a discussion on the connections with Gaussian covariance estimation).

\paragraph{Lower Bounds in DP.} Even before the inception of DP, researchers had devised lower bounds on the accuracy of algorithms that avoid data re-identification~\citep{DinurN03}. Since then there has been a long line of work on lower bounds for DP algorithms, such as packing arguments for pure (i.e., \((\eps, 0)\)-) DP~\citep{HardtU14}, reconstruction arguments using discrepancy theory~\citep{MuthukrishnanN12}, information theoretical tools~\citep{AcharyaSZ21}, or fingerprinting techniques.

\paragraph{Fingerprinting Techniques.}    Fingerprinting codes~\citep{BonehS98a,Tardos08a} were first used in DP by~\citet{BunUV18a} to prove lower bounds for answering counting queries on data from the hypercube. Several works then built upon these ideas to obtain lower bounds for approximate DP algorithms for a variety of problems: statistical queries \citep{HardtU14, SteinkeU15a}, private subspace estimation \citep{DworkTTZ14a}, mean of vectors with \(\pm 1\) entries \citep{SteinkeU16a}, and other problems under very weak accuracy guarantees~\citep{PeterTU23a}. A problem with fingerprinting codes is that they usually have a non-trivial construction and, because of that, are used nearly as black-boxes. \citet{DworkSSUV15a} was one of the first works to boil down the main techniques to simpler calculations on the expected value of some ``correlation statistics''.
This general strategy was later called the ``Fingerprinting Lemma'' by~\citet{BunSU17a}. \citet{PeterTU23a} thoroughly discusses the main differences between fingerprinting codes and fingerprinting lemmas.

\paragraph{Beyond i.i.d.~Priors in Fingerprinting Lemmas.}  
Many fingerprinting techniques in previous works required randomizing the parameters \emph{independently} of each other, which is not well suited for more structured problems such as covariance estimation. \citet{KamathMS22a} and \citet{CaiWZ23a} propose fingerprinting arguments that try to circumvent this limitation. \citet{KamathMS22a} proposes a generalized fingerprinting lemma for parameter estimation of exponential family distributions and show tight lower bounds for Gaussian covariance estimation. Ultimately they still require a distribution with bounded support and independent entries over the \emph{natural parameters} of the exponential family parameterization of the problem, which restricts their results to the high accuracy regime defined earlier.

\citet{CaiWZ23a} proposes the \emph{score attack}, a fingerprinting-type argument for parameter estimation of any probability distribution with a differentiable density. The core of their framework is the use of the Stein identity \citep[Proposition~2.2]{CaiWZ23a}. 
Their use of the identity claims to hold for any distribution on the parameters. However, we found a gap in their proof, and their results \emph{only hold when the parameters are independent} (confirmed in private communication with the authors). 
Interestingly, many of previous fingerprinting lemmas used results analogous to the Stein identity for specific 
(sometimes discrete) distributions (e.g., \citealp[Lemma~14]{DworkTTZ14a} or \citealp[proof of Lemma~6.8]{Kamath0SU19a}).

We build upon the score attack (using a few of their results whose correctness we could verify) to show lower bounds for DP covariance estimation. We use a distribution over symmetric matrices with \emph{unbounded support} and \emph{non-independent} parameters. Key to our analysis is a generalization of Stein's identity known as the Stein-Haff identity~\citep{Haff79a}. Moreover, our analysis technique suggests a general strategy to handle general prior distributions over the parameters in the score attack using Stokes' Theorem. We discuss these connections in Section~\ref{sec:LB_haff}.

During the development of our work, \citet{Narayanan23a} proved an \(n = \tilde{\Omega}(d^2/\eps \alpha)\) lower bound in the low-accuracy ($\alpha = O(\sqrt{d})$) regime when \(\delta = O(1/d^2)\).
They elegantly use that the inverse Wishart distribution is the conjugate prior of the Wishart distribution, but the proof requires \(\alpha = \Omega(1)\) and other reductions that worsen the lower bound by polylog factors.
Our result applies to the more general setting of $\delta$ as large as $O(1/n\ln n)$, which is nearly as large as possible\footnote{
    The regime $\delta \geq 1/n$ is typically not used for DP, otherwise blatantly non-private mechanisms exist.
}. Although both our work and~\citet{Narayanan23a} use fingerprinting-type arguments, the technical approaches are completely unrelated.

\subsection{Notation}

Throughout the paper, we write \(a \lesssim b\) if there is some universal constant \(C > 0\) such that \(a \leq C b\). We denote by \(\Sym{d}\) and \(\PSDCone\) the set of symmetric and the set of positive semi-definite \(d \times d\) matrices, respectively, and the Frobenius norm by \(\norm{\cdot}_F\).
We let $\expect{\cdot}$ denote an unconditional expectation, and $\ESigma{\cdot}$ denote an expectation conditioned on $\Sigma$.
The notation $\prob{\cdot}$ and $\PrSigma[\cdot]$ are defined analogously.

\section{DP Lower Bounds via the Score Attack}
\SectionName{Overview}

Our results build upon the \emph{score attack} framework~\citep{CaiWZ23a} to prove lower bounds on the error  of parameter estimation of probability distributions under approximate DP.
Cai et al.~apply their framework to models that are prevalent in statistics, such as GLMs, whereas we focus on the fundamental model of a single multivariate Gaussian with unknown covariance, for which our understanding is still not complete.

Let us use the setting of Gaussian covariance estimation to discuss more in depth fingerprinting arguments and the framework of the score attack. Throughout the paper, let \(x_1, \dotsc, x_n\) be i.i.d. random vectors in \(\bR^d\) with distribution \(\cN(0, \Sigma)\) and density \(p(\cdot \cond \Sigma)\), where \(\Sigma \in \PSDCone\). Let \(\cM\) be a mechanism that estimates \(\Sigma\) when given as input the matrix \(X \in \bR^{d \times n}\) whose columns are \(x_1, \dotsc, x_n\).

The main intuition behind fingerprinting arguments is that, if \(\cM(X)\) predicts \(\Sigma\) somewhat accurately, it should have some correlation with \(x_1, \dotsc, x_n\)
This intuition is not true in general since the mechanism that always outputs \(\Sigma\) is perfectly accurate and completely independent of its input. Yet, if \(\Sigma\) is unknown to the mechanism
(e.g., if it is chosen randomly in the right way), this intuition can often be formalized. 

The argument a quantity \(\cA(z, \cM(X))\), a ``correlation statistic''  of \(z \in \bR^d\) with \(\cM(X)\). It should have the property that, for random \(\Sigma\),
\begin{enumerate}[(i)]
    \item \(\expect{\abs{\cA(z, \cM(X))}}\) is small if \(z\) is independent of \(X\) and \(\cM\); \label{item:fpl_ii}
    \item \(\expect{\cA(z, \cM(X))}\) is large if \(z = x_i\) for a uniformly random \(i \in [n]\); \label{item:fpl_i}
\end{enumerate}
 Property~\eqref{item:fpl_ii} is usually guaranteed by the design of \(\cA\). Property~\eqref{item:fpl_i} often requires a more careful analysis and choice of the distribution of \(\Sigma\), and this is the property usually called ``Fingerprinting Lemma'' in earlier works. None of the above properties depend on differential privacy: DP comes into play to show that, in fact, the expected statistics in both cases above are close to each other. Intuitively, if \(\cM\) is differentially private, it cannot be too correlated with \(x_i\) for any \(i \in [n]\).
In this paper we will barely even use the definition of $(\epsilon,\delta)$-DP because it is only used inside \Theorem{Cai}, which we use as a black box.

The score attack framework of \citet{CaiWZ23a} proposes the use of the statistic
\begin{equation}
    \EquationName{Statistic}
    \cA(z, \cM(X)) \coloneqq
    \iprod{\cM(X) - \Sigma}{\nabla_{\Sigma} \ln p(z \cond \Sigma)},
\end{equation}
where\footnote{The score function in this case should be the gradient of \(p\) as a function of only the lower triangular entries of \(\Sigma\) to account for symmetry. Yet, we can in this case use the matrix gradient. These details are important as well when discussing the Fisher information. We defer the details to~\Appendix{ScoreFisherAppendix}.} \(\nabla_{\Sigma} \ln p(z \cond \Sigma) \in \PSDCone\) is known as the score function in the statistics literature. The next theorem summarizes  the analysis framework of \citet{CaiWZ23a}. Note that in the next theorem, $T$ is a parameter that can be optimized to improve the upper bound.

\begin{theorem}[Framework for the score attack, {\protect\citealp{CaiWZ23a}}]
    \TheoremName{Cai}
    For all \(T\) and random \(\Sigma \in \PSDCone\),
    \begin{equation*}
        \sum_{i = 1}^n \ESigma{\cA(x_i, \cM(X))}~\leq~ 
        \sum_{i=1}^n \Big( 2 \eps \alpha_{\Sigma} \sqrt{\lambda_{\max}(\cI(\Sigma))} + 2 \delta T  + \int_{T}^\infty \PrSigma[\abs{\cA(x_i,\cM(X))} \geq t] \diff t \Big),    
    \end{equation*}
    where \(\alpha_{\Sigma} \coloneqq \ESigmaSmall{\norm{\cM(X) - \Sigma}_F^2}^{1/2}\)
     and {\(\cI(\Sigma)\)} is the Fisher information matrix of \(p(\cdot \cond \Sigma)\). Moreover,
     \begin{equation}
        \label{eq:score_to_divergence}
        \sum_{i = 1}^n \ESigma{\cA(x_i, \cM(X))}
        ~=~ \sum_{i,j \in [d] \colon i \geq j } \frac{\partial}{\partial \Sigma_{ij}} g(\Sigma)_{ij}, 
        \qquad \text{where}~g \coloneqq \ESigma{\cM(X)}.
     \end{equation}
     
\end{theorem}

The first inequality is useful to show that score attack statistics in~\eqref{item:fpl_ii} are roughly upper bounded by \(\alpha_{\Sigma} \sqrt{\lambda_{\max}(\cI(\Sigma))}\) and uses that scenarios~\eqref{item:fpl_ii} and~\eqref{item:fpl_i} are not too far apart. The second equation roughly shows that~\eqref{item:fpl_i} will be \(\Omega(d^2)\) if \(g(\Sigma) = \ESigma{\cM(X)} \approx \Sigma\) (the right-hand side of~\eqref{eq:score_to_divergence} is exactly \(d(d+1)/2\) if \(g(\Sigma) = \Sigma\)). Yet, one needs to carefully pick the distribution of \(\Sigma\) to formalize this intuition.

\subsection{Proof Overview}
\SectionName{ProofOverview}

Let us overview how the framework above can be used to prove \Theorem{Main}. Note that
\Theorem{Cai} does not require \(\Sigma\) to be random. The challenge is to randomize \(\Sigma\) in a way that we can meaningfully lower bound the right-hand side of~\eqref{eq:score_to_divergence}. We will choose $\Sigma$ to have a Wishart distribution with appropriate parameters, then show both an upper and a lower bound on the expected statistics
$\sum_{i = 1}^n \expect{\cA(x_i, \cM(X))}$.
\begin{itemize}
    \setlength{\topsep}{0pt}\setlength{\itemsep}{0pt}\setlength{\parsep}{0pt}\setlength{\parskip}{0pt}\setlength{\partopsep}{0pt}
    \item \textit{A lower bound}
    of \(\Omega(d^2)\) is proven in    
    Lemma~\ref{lemma:lower_bound_covariance}. 
    
    \item \textit{An upper bound} of roughly $O(n \epsilon \alpha)$ is proven in Lemma~\ref{lemma:upper_bound_covariance_randomized_sigma}.
\end{itemize} 
The complete proof of \Theorem{Main}, given in \Section{putting_everything_together}, straightforwardly combines these bounds.

\subsection{Distribution on $\Sigma$}
\SectionName{DistrSigma}

As discussed as the beginning of this section, for the lower bound in \Theorem{Cai} (or any fingerprinting argument) to be non-trivial we need to carefully select a distribution on the covariance matrix \(\Sigma\). We will use one of the most natural distributions over $\PSDCone$, the (normalized) Wishart distribution.

\paragraph{The normalized Wishart Distribution.} Let \(G\) be a \(d \times D\) random standard Gaussian matrix, and let 
\begin{equation}
    \label{eq:def_sigma}
    \Sigma \coloneqq \frac{1}{D} G G^{\transp}
    \quad \text{with}\quad D ~=~ 2 d.
\end{equation}
The distribution of \(G G^{\transp}\) is known as the \emph{Wishart distribution} (of dimension \(d\)) with \(D\) degrees of freedom.
We refer to the distribution of $\Sigma$ as above as the \emph{normalized Wishart distribution} with \(D\) degrees of freedom. 
The choice \(D = 2d\) is to ensure that $\Sigma$ has constant condition number with high probability.

Although natural, this distribution was not used in previous fingerprinting arguments. \citet{KamathMS22a} proposes the Generalized Fingerprinting Lemma (for exponential families). As stated, it requires the distribution of each of the coordinates of \(\Sigma^{-1}\) to be independent and uniform over a bounded interval, which already rules out a Wishart distribution, even if truncated to be bounded. This also forces the distribution to be such that the diameter \(\expectsmall{\norm{\Sigma - \expect{\Sigma}}_F^2}\) is \(O(1)\), which makes their bounds only hold on the high-accuracy regime. \citet{Narayanan23a} uses a Wishart distribution for the \emph{inverse covariance matrix} \(\Sigma^{-1}\), which has diameter \(\Theta(d)\). However, their analysis requires \(\delta = O(1/d^2)\) to get a tight upper bound on his correlation statistics (which are not the statistics $\cA$ defined in \eqref{eq:Statistic}). In our case, we also want a distribution with diameter \(\Theta(d)\) for which, simultaneously, we can meaningfully lower bound the expected value of~\eqref{eq:score_to_divergence}. As we shall see, the choice of the Wishart distribution leads to an elegant analysis.

\paragraph{Error of the Mechanism.} Let \(\WellCond \coloneqq \setst{A \in \PSDCone}{0.09 I \preceq A \preceq 10 I}\) be the set of what we shall call \emph{well-conditioned matrices}. Define the expected error \(\alpha^2\) of \(\cM\)  by
\begin{equation}
    \EquationName{alphadef}
\alpha^2 \coloneqq \expectg{\alpha_\Sigma^2}{\Sigma \in \WellCond}
\qquad \text{where}\qquad 
\alpha_\Sigma^2 \coloneqq \ESigma{\norm{\cM(X) - \Sigma}_F^2}.    
\end{equation}
Readers familiar with the Mahalanobis norm will note that, under the event \(\Sigma \in \WellCond\), the error in Mahalanobis norm is the same as \(\norm{\cM(X) - \Sigma}_F\) up to constants. 
Thus, lower bounds using \(\alpha\) imply lower bounds on mechanisms with guarantees under the Mahalanobis norm.

To use the lower bound from the framework of \Theorem{Cai},
the next lemma (the ``Fingerprinting Lemma'' of our argument) will show non-trivial lower bounds for \(\Sigma\) as described in~\eqref{eq:def_sigma}.

\begin{lemma}[Main Lower Bound]
    \label{lemma:lower_bound_covariance} 
    If \(\alpha \leq \sqrt{d}/15\) and \(d \geq 20\), then
    \begin{equation*}
        \sum_{i = 1}^n \expect{\cA(x_i, \cM(X))}
        ~\geq~ \frac{d^2}{4}.
    \end{equation*}
\end{lemma}

The main technical step in applying the upper bound \Theorem{Cai} is as follows.

\begin{lemma}[Main Upper Bound, Fixed $\Sigma$]
    \label{lemma:upper_bound_covariance}
    Let \(\beta_u > 0\) be an arbitrary constant. 
    Assume that $\delta$ satisfies
    \eqref{eq:DeltaAssumption}, that
    \(\cM(X)\) is \((\eps, \delta)\)-DP, and that \(\cM(X) \preceq \beta_u I\).
    Then, for sufficiently large $n$,
    \begin{equation*}
        \sum_{i = 1}^n \ESigma{\cA(x_i, \cM(X))}
        ~\leq~ 2 n \eps  \frac{\alpha_{\Sigma}}{\lambda_{\min}(\Sigma)}   + \paren*{\frac{\beta_u}{\lambda_{\min}(\Sigma)} + 1}d^{3/2}.
    \end{equation*}
\end{lemma}

\subsection{Completing the proof of \Theorem{Main}}
\SectionName{putting_everything_together}

Our upper bound (Lemma~\ref{lemma:upper_bound_covariance}) holds for all $\Sigma$.
To obtain a bound depending only on the problem parameters ($n$, $d$, etc.) we will let $\Sigma$ follow a normalized  Wishart distribution as described above.

\begin{lemma}[Upper Bound with Random $\Sigma$]
\label{lemma:upper_bound_covariance_randomized_sigma}
    Assume \(\cM(X)\) is \((\eps, \delta)\)-DP with \(\delta \leq 1/3n \ln(en)\) and that \(\cM(X) \preceq 10 I\). 
    Then there is a constant \(C > 0\) such that, for large enough \(n\),
    \begin{equation*}
        \sum_{i = 1}^n \expect{\cA(x_i, \cM(X))} ~\leq~ C n \eps (\alpha + 2^{-d})  + 66 \cdot d^{3/2}.
    \end{equation*}
\end{lemma}

We give a sketch below, and a complete proof in \Appendix{OmittedOverview}.

\begin{sketch}
Taking the expectation with respect to \(\Sigma\) on the inequality of Lemma~\ref{lemma:upper_bound_covariance} yields\footnote{Here we are implicitly using the tower property of conditional expectations, that $\expect{\ESigma{\cdot}}=\expect{\cdot}.$}
    \begin{equation*}
        \sum_{i = 1}^n \expect{\cA(x_i, \cM(X))} \leq 2 n \eps  \underbrace{\expect{\frac{\alpha_{\Sigma}}{\lmin(\Sigma)}}}_{O(\alpha + 2^{-d})}
        +
        \underbrace{\paren*{10\expect{\frac{1}{\lmin(\Sigma)}} + 1}}_{O(1)} d^{3/2}.
    \end{equation*}
    For the second term, a standard bound for Wishart matrices is $\expect{\lmin^{-1}(\Sigma)} \leq 6.5$
    (see Lemma~\ref{lemma:bound_expected_inv_lambda_min}).
    
    The first term requires care because $\alpha_\Sigma$ is random since it depends on $\Sigma$.
    Most of the contribution of this term is on the event $\cE \coloneqq \curly{\Sigma \in \WellCond}$ since \(\Pr[ \Sigma \not\in \WellCond]\) is exponentially small. On the event $\cE$ we have $1/\lmin(\Sigma) = O(1)$,
    so the contribution is $O(1)\cdot \expect{\ones_{\cE} \cdot \alpha_\Sigma} \leq O(1) \cdot \sqrt{\expect{\ones_{\cE} \cdot \alpha_\Sigma^2}} = O(\alpha)$.
\end{sketch}

Combining our main lower bound (Lemma~\ref{lemma:lower_bound_covariance})
and main upper bound (Lemma~\ref{lemma:upper_bound_covariance_randomized_sigma}), it is straightforward to obtain our main theorem.

\repeatclaim{\Theorem{Main}}{
    Let \(\cM \colon (\bR^d)^n \to \PSDCone\) be \((\eps, \delta)\)-DP where  \(\eps \in (0,1)\) and 
    $\delta \leq 1/3 n \ln(en)$. 
    Suppose that $\alpha \in [2^{-d},\frac{\sqrt{d}}{15}]$.
    Then
    \[
    n ~=~ \Omega\Big(\frac{d^2}{\eps \alpha}\Big).
    \]
}
\begin{proof}
    Without loss of generality, we may assume that the output of $\cM$ lies in $\WellCond$ by projecting the output $\cM(X)$ onto $\WellCond$. Doing so does not increase $\alpha_\Sigma$ for any $\Sigma \in \WellCond$ (and hence does not increase $\alpha$) since projection onto convex sets does not increase the Euclidean distance, and the Frobenius norm is a Euclidean norm.
    
    Let $\Sigma$ have the Wishart distribution described in \Section{DistrSigma}.
    Then  Lemmas~\ref{lemma:lower_bound_covariance} and~\ref{lemma:upper_bound_covariance_randomized_sigma} imply that
    \begin{equation*}
        \frac{d^2}{5} ~\leq~ \sum_{i = 1}^n \expect{\cA(x_i, \cM(X))}
        ~\leq~ C n \eps (\alpha + 2^{-d})  + 66 d^{3/2}
        ~\leq~ 2 C n \eps \alpha + 66 d^{3/2},
    \end{equation*}
    since $\alpha \geq 2^{-d}$.
    Rearranging, we obtain
    $n \geq (d^2/5-66d^{3/2})/2 C \epsilon \alpha$,
    which is $\Omega(d^2/\eps \alpha)$ as required.
\end{proof}

\section{Lower Bound on the Correlation Satistics via the Stein-Haff Identity}
\label{sec:LB_haff}

    In this section we shall prove our main lower bound on the correlation statistics, formally stated in Lemma~\ref{lemma:lower_bound_covariance}.
    The main idea is to lower bound \eqref{eq:score_to_divergence} using a result known as the Stein-Haff identity, which extends Stein's identity for Gaussian random variables to Wishart matrices.

    First, we require some notation. Define the \(d \times d\) matrix of differential operators \(\bD_{\Sigma}\) by
\begin{equation*}
    \bD_{\Sigma}(i,j) = \frac{(1 + \boole{i = j})}{2} \cdot  \frac{\partial }{\partial \Sigma_{ij}}. \qquad \forall i,j \in [d],
\end{equation*}
where we set $\boole{P}$ to be 1 if the predicate \(P\) is true, and 0 otherwise. Crucially, we identify \(\Sigma_{ij}\) and \(\Sigma_{ji}\) when differentiating. In other words, we see any function of a symmetric matrix \(\Sigma\) as a function of its lower triangular entries. It may not be widely known, but this operator is the one that leads to the correct definition of a gradient over \(\PSDCone\) (see~\citealt{SrinivasanP23a} or \Appendix{SymmetricGradient}). Surprisingly, even in prominent parts of the literature there is disagreement about the proper notion of a gradient that takes into account matrix symmetry.  In this paper we treat these details carefully with the hope that it is instructive to the reader. For this section, it helps to note that if \(g \colon \PSDCone \to \Reals\) is differentiable, then \(\iprod{\bD_{\Sigma}}{g(\Sigma)} = \sum_{i \geq j} \partial_{\Sigma_{ij}} g(\Sigma)_{ij}\) is the \emph{divergence} of \(g\) as a function of the lower triangular entries of \(\Sigma\).

The next theorem is an extension of the classical Stein's identity  \citep{Stein72a, SteinDHR04a} from normal random variables to Wishart random variables, and we shall state it in terms of general Wishart distributions. That is, we say that \(\Sigma \sim \cW_d(D; V)\) for non-singular \(V \in \PSDCone\) if \(\Sigma = G G^{\transp}\) where \(G\) is a \(\bR^{d \times D}\) matrix whose columns are i.i.d.\ vectors each with distribution \(\cN(0, V)\).

\begin{theorem}[{Stein-Haff Identity, \citealp[Thm.~2.1]{Haff79a}}]
\label{thm:haff_identity}
Assume \(g \colon \PSDCone \to \Reals\) satisfy some mild regularity conditions (see~Appendix~\ref{apx:stein-haff_conditions}),
and let \(\Sigma \sim \cW_d(D; V)\) for some non-singular \(V \in \PSDCone\). Then
\begin{equation*}
    \expect{\iprod{\bD_{\Sigma}}{g(\Sigma)}}
    = \frac{1}{2} \expect{\iprod{V^{-1} - \paren{D - d - 1} \Sigma^{-1}}{g(\Sigma)}}.
\end{equation*}    
\end{theorem}

The original proof of this identity uses Stokes' Theorem and using the fact that
\begin{equation}
\EquationName{DerivativeOfDensity}
\bD_\Sigma \cdot p_{\cW}(\Sigma) ~=~ \tfrac{1}{2}(V^{-1} - (D - d - 1)\Sigma^{-1}) p_{\cW}(\Sigma)
\end{equation}
where \(p_{\cW}\) is the density of the Wishart distribution. The high level idea is to handle the left-hand side with integration by parts, moving the differential operator from \(g\) to the density of the Wishart distribution.  This can be seen as a direct generalization of the integration by parts proof of the classical Stein's identity. A proof can be found in~\citet{Haff79a} and, using modern notation and tools, in~\citet[\S 5]{TsukumaK20a}. Furthermore, this suggests a general avenue to prove lower bounds on the score statistics even when the parameters are not independent: use Stokes' theorem to write the expected divergence into an expression the depends on the gradient of the density, and then manipulate this expression to connect \(g(\Sigma)\) to the accuracy of \(\cM\). That is exact what we do in Lemma~\ref{lemma:haff_to_accuracy}.

\begin{lemma}
    \label{lemma:haff_to_accuracy}
    Let \(g \colon \PSDCone \to \Reals\) be differentiable and let \(\Sigma \sim \cW_d(D; V)\) for a non-singular \(V \in \PSDCone\). Then
    \begin{equation*}
        \expect{\iprod{\bD_{\Sigma}}{g(\Sigma)}}
    \geq \frac{d (d+1)}{2} - \frac{1}{2}\sqrt{\expect{\norm{\Sigma - g(\Sigma)}_F^2}} \cdot \sqrt{\expect{\norm{(D - d - 1)\Sigma^{-1} - V^{-1}}_F^2}}.
    \end{equation*}
    In particular, if \(V = \frac{1}{D} I\) and \(D = 2d\) we have
    \begin{equation*}
        \expect{\iprod{\bD_{\Sigma}}{g(\Sigma)}}
    \geq \frac{d (d+1)}{2} - 2 d^{1.5} \sqrt{\expect{\norm{\Sigma - g(\Sigma)}_F^2}}.
    \end{equation*}
\end{lemma}
\begin{proof}
    By the Stein-Haff identity, we have
    \begin{align*}
        \expect{\iprod{\bD_{\Sigma}}{g(\Sigma)}}
    &= \frac{1}{2} \expect{\iprod{V^{-1} - \paren{D - d - 1} \Sigma^{-1}}{g(\Sigma)}}
        \\
        &= \frac{1}{2} \expect{\iprod{V^{-1} - \paren{D - d - 1} \Sigma^{-1}}{g(\Sigma) - \Sigma} + \iprod{V^{-1} - \paren{D - d - 1} \Sigma^{-1}}{\Sigma}}.
    \end{align*}
    Using the fact that \(\expect{\Sigma} = D V\) we have
    \begin{align*}
        \expect{\iprod{V^{-1} - \paren{D - d - 1} \Sigma^{-1}}{\Sigma}}
        &= \expect{\Tr(\Sigma V^{-1} - (D - d -1)I)}
        \\
        &= \Tr(D I - (D - d - 1) I) = d (d + 1).
    \end{align*}
    Finally, the desired inequality follows since, by Cauchy-Schwartz,
    \begin{align*}
        &\expect{\iprod{V^{-1} - \paren{D - d - 1} \Sigma^{-1}}{g(\Sigma) - \Sigma}}
        \\
        \geq&        
        -  \sqrt{\expect{\norm{\Sigma - g(\Sigma)}_F^2}} \sqrt{\expect{\norm{(D - d - 1)\Sigma^{-1} - V^{-1}}_F^2}}.
    \end{align*}
    Moreover, since \(\Sigma\) follows a Wishart distribution, many properties of \(\Sigma^{-1}\), such as the expectation and variance of its entries, are well known. (See Lemma~\ref{lemma:properties_inv_wishart}.)
    Specifically, \(\expect{\Sigma^{-1}} = \frac{1}{D - d - 1} V^{-1}\) and, thus,
    \begin{align*}
        \expect{\norm{(D - d - 1)\Sigma^{-1} - V^{-1}}_F^2}
        &=
        (D - d - 1)^2 \sum_{i,j \in [d]} \expect{(\Sigma_{ij}^{-1} - \tfrac{1}{D -d - 1} V_{ij}^{-1})^2}
        \\
        &= (D - d - 1)^2 \sum_{i,j \in [d]} \var{\Sigma_{ij}^{-1}}.
    \end{align*}
    Now consider the case that \(V = \frac{1}{D}I\) with \(D = 2d\). Then \(V_{ij}^{-1} = \boole{i = j} D\) for all \(i,j \in [d]\). Combined with the variance formulas from~\ref{lemma:properties_inv_wishart}, we get
    \begin{align*}
        (D - d - 1)^2  \sum_{i,j \in [d]} \var{\Sigma_{ij}^{-1}}
        &= d \cdot  \frac{
            2 D^2}{(D - d - 3)} + d(d-1) \cdot \frac{
                (D - d -1)D^2}{(D - d - 3) (D-d)}
        \\
        &= d \cdot  \frac{
            2 \cdot 4d^2 }{(d - 3)} + d(d-1) \cdot \frac{
                (d -1)4d^2}{(d - 3) d}
        \\
        &= 4 d^3 \paren*{\frac{
            2}{(d - 3)} +  \frac{
                (d -1)^2}{(d - 3) d}}
        \leq 16 d^3,
    \end{align*}
    since we assume that \(d \geq 5\).  
    This inequality holds since the left-hand side is decreasing for $d\geq 5$: the derivative of $\frac{2}{(d - 3)} +  \frac{(d -1)^2}{(d - 3) d}$
    is $\frac{-3d^2-2d+3}{(d-3)^2 d^2}$, and the numerator is negative for $d \geq 5$.
\end{proof}

We are now in place to prove the main lower bound (Lemma~\ref{lemma:lower_bound_covariance}).

\begin{proofof}{Lemma~\ref{lemma:lower_bound_covariance}}
    Define \(g(\Sigma) = \ESigma{\cM(X)}\).  By Theorem~\ref{thm:Cai}, together with the lower bound proven in Lemma~\ref{lemma:haff_to_accuracy}, we have 
    \begin{align*}
        \sum_{i = 1}^n \expect{\cA(x_i, \cM(X))} 
        &= \sum_{i,j \in [d] \colon i \geq j} \expectBig{\frac{\partial}{\partial \Sigma_{ij}} g(\Sigma)_{ij}} = \expect{\iprod{\bD_\Sigma}{g(\Sigma)}}
        \\
        & \geq \frac{d (d+1)}{2} - 2 d^{1.5} \sqrt{\expect{\norm{\Sigma - g(\Sigma)}_F^2}}.
    \end{align*}
    Note that by the conditional Jensen's inequality and the definition of \(\alpha_{\Sigma}^2\) from~\eqref{eq:alphadef} we have
    \begin{equation*}
        \expectsmall{\norm{g(\Sigma) - \Sigma}_F^2} 
        \leq \expectsmall{\norm{\cM(X) - \Sigma}_F^2} = \expect{\alpha_{\Sigma}^2}.   
    \end{equation*}
     By assumption, \(d \geq 20\) and \(\alpha^2 \leq d/(15)^2 \leq d/200\). Thus, by Lemma~\ref{lemma:expect_sigma_squared} we have
    \begin{equation*}
        \expect{\alpha_{\Sigma}^2} \leq \alpha^2 + \frac{d}{200} \leq \frac{d}{100}.
    \end{equation*}
    Combining the above facts yields 
    \begin{equation*}
        \frac{d (d+1)}{2} - 2 d^{1.5} \sqrt{\expect{\norm{\Sigma - g(\Sigma)}_F^2}} \geq \frac{d(d +1)}{2} - \frac{2 d^2}{10}
        \geq \frac{3 d^2}{10} \geq \frac{d^2}{4}, 
    \end{equation*}
    as desired. 
\end{proofof}

\section{Upper Bound on the Correlation Satistics}
\SectionName{UB}

This section proves Lemma~\ref{lemma:upper_bound_covariance}, an upper-bound on the correlation statistics when \(\cM\) is differentially private. These bounds are parameterized the covariance matrix \(\Sigma\) and, thus, are random variables. We shall do so by using the upper bound of \Theorem{Cai}, which depends on the tails of the correlation statistics. We give tail bounds for this statistics in the following lemma.

\begin{lemma}
    \label{lemma:tail_bound_attack_statistic}
    Assume \( \cM(X) \preceq \beta_u I\) and set \(\gamma \coloneqq 1+\beta_u/\lmin(\Sigma)\).
    Then, for any \(i \in [n]\) and \(T \geq 6 \gamma d^{3/2}\),
    \begin{equation*}
        \int_{T}^{\infty} \PrSigma(\abs{\cA(x_i, \cM(X))} \geq t) \diff t ~\leq~ 9 \gamma \sqrt{d} \exp\paren*{-\frac{T}{9 \gamma \sqrt{d}}}.
    \end{equation*}
\end{lemma}
\begin{proof}
    To simplify notation, let
    $A_i=\abs{\cA(x_i, \cM(X))}$ and 
    $z \coloneqq \Sigma^{-1/2} x_i$.
    Note that $z$ has the distribution $\cN(0,I)$.
    Then,
    \begin{alignat}{2}\nonumber
        A_i
          &~=~ \abs{\iprod{\cM(X) - \Sigma}{ \nabla_{\Sigma}\, p(x_i \cond \Sigma)}}
          &&\quad\text{(by definition in \eqref{eq:Statistic})}
          \\\nonumber
          &~=~ \abs{\iprod{\cM(X) - \Sigma}{\Sigma^{-1} x x^{\transp} \Sigma^{-1} - \Sigma^{-1}}}
          &&\quad\text{(by Proposition~\ref{prop:Score})}
          \\\nonumber
          &~=~ \abs{\iprod{\Sigma^{-1/2}\cM(X)\Sigma^{-1/2} - I}{ z z^{\transp} - I}}
          &&\quad\text{(cyclic property of trace)}
          \\\nonumber
          &~\leq~ \norm{\Sigma^{-1/2}\cM(X)\Sigma^{-1/2} - I}_F \cdot \norm{z z^{\transp} - I}_{F}
          \\\EquationName{MatrixCauchy}  
          &~\leq~
          \sqrt{d} \underbrace{\paren*{\frac{\beta_u}{\lambda_{\min}(\Sigma)} + 1}}_{=\gamma} \cdot \norm{z z^{\transp} - I}_{F},
    \end{alignat}
    using the triangle inequality and $\Sigma^{-1/2}\cM(X)\Sigma^{-1/2} \preceq \beta_u \Sigma^{-1} \preceq \frac{\beta_u}{\lmin(\Sigma)} I $.

    To complete the proof, it suffices to prove a tail bound on \eqref{eq:MatrixCauchy}.
    To begin, observe that
    $z z^{\transp}-I$ has eigenvalues $\norm{z}_2^2-1$ with multiplicity $1$, and $-1$ with multiplicity $d-1$, so
    \[
    \norm{z z^{\transp} - I}_{F}
    ~=~ \sqrt{(\norm{z}_2^2 -1)^2 + d-1}
    ~\leq~ \sqrt{\norm{z}_2^4 + d}.
    \]
    Therefore,
    \begin{align*}
        \PrSigma[\, A_i \geq t \,] 
        &~\leq~ \PrSigma\Big[\gamma \sqrt{d} \cdot \sqrt{\norm{z}_2^4 + d} \geq t \Big]
        ~=~
        \PrSigma\Bigg[\norm{z}_2^2 \geq \sqrt{\frac{t^2}{\gamma^2 d} -d } \Bigg].
    \end{align*}
    For any $t \geq 6 \gamma d^{3/2}$, we may write
    $t \geq t/2 + 6 \gamma d^{3/2}/2$. Squaring  then dividing by $\gamma^2 d$, we get
    \[
    \frac{t^2}{\gamma^2 d} ~\geq~ \frac{t^2}{4 \gamma^2 d} + \frac{36 d^2}{4}
    ~\geq~ \frac{2}{9} \cdot \frac{ t^2}{\gamma^2 d} + 8 d^2 + d
    ~=~ 18 x^2 + 8 d^2 + d,
    \]
    where we have defined $x=t/9\gamma \sqrt{d}$.
    Thus, we have
    \[
    \PrSigma[\, A_i \geq t \,]
    ~\leq~
    \PrSigma\Bigg[\norm{z}_2^2 \geq
    \sqrt{\frac{t^2}{\gamma^2 d} -d } \Bigg]
    ~\leq~
    \PrSigma\Bigg[\norm{z}_2^2 \geq \sqrt{8d^2+18x^2} \Bigg]
    ~\leq~ e^{-x},
    \]
    by Corollary~\ref{cor:chi_squared}.    
    We assume $T \geq 6 \gamma d^{3/2}$, so
    \begin{equation*}
        \int_{T}^{\infty} \PrSigma[\, A_i \geq t \,] \diff t
        \leq \int_{T}^{\infty} \exp\Big(-\frac{t}{9 \gamma \sqrt{d}}\Big) \diff t = 9 \gamma \sqrt{d} \exp\Big(-\frac{T}{9 \gamma \sqrt{d}}\Big),
    \end{equation*}
as desired.
\end{proof}

We are now in position to prove Lemma~\ref{lemma:upper_bound_covariance}.

\vspace{6pt}

\begin{proofof}{Lemma~\ref{lemma:upper_bound_covariance}}
    Let \(z \sim \cN(0, \Sigma)\) be independent of \(X\) and \(X_i'\) be identical to the matrix \(X\) except with its \(i\)-th column replace by \(z\). By Theorem~\ref{thm:Cai} we have, for any \(T > 0\), 
    \begin{equation*}
        \sum_{i = 1}^n \ESigma{\cA(x_i, \cM(X))}
        \leq 2 \eps n  \alpha_{\Sigma} \sqrt{\lambda_{\max}(\cI(\Sigma))} + 2 n\delta T + \sum_{i = 1}^n \int_{T}^\infty \PrSigma[A_i \geq t]\diff t
    \end{equation*}
    where, as before, we let $A_i=\abs{\cA(x_i, \cM(X))}$.
    Let us first bound the latter 2 terms in the right-hand side. Set \(T \coloneqq 9 \gamma d^{3/2} \ln(1/\delta)\) where \(\gamma \coloneqq \beta_u/\lambda_{\min}(\Sigma) + 1\). Since \(\delta \leq 1/e\), we have \(T \geq 6 \gamma d^{3/2}\).
    Thus, by Lemma~\ref{lemma:tail_bound_attack_statistic} we have
    \begin{align*}
        2 n \delta T + \sum_{i = 1}^n \int_{T}^\infty \PrSigma[A_i \geq t]\diff t
        &~\leq~ 9 \gamma  n \paren*{2 d^{3/2} \delta\ln(1/\delta) + \sqrt{d}\exp\big(-d\ln(1/\delta)\big)}
        \\
        &~\leq~ 
        18 \gamma n d^{3/2} \paren*{ \delta \ln(1/\delta) + \delta} ~\leq~ 
        36 \gamma d^{3/2}
    \end{align*}
    since $\delta \ln(1/\delta) \leq 1/n$ for all $n \geq 1$, by our choice of $\delta$ in \eqref{eq:DeltaAssumption}.

    To complete the proof of the desired inequality, note that 
    Lemma~\ref{lemma:fisher_information} yields
    \begin{equation*}
        2 \eps n  \alpha_{\Sigma} \sqrt{\lambda_{\max}(\cI(\Sigma))} 
        \leq 2 \eps n \frac{\alpha_{\Sigma}}{\lambda_{\min}(\Sigma)},
    \end{equation*} 
    as desired.
\end{proofof}


\clearpage
\appendix
 
\section{Conditions for the Stein-Haff Identity}
\label{apx:stein-haff_conditions}

Let us now describe the conditions a function $g \colon \PDCone \to \PSDCone$ ought to satisfy for the Stein-Haff identity (Theorem~\ref{thm:haff_identity}) to hold, where \(\PDCone\) is the set of positive definite matrices. These are fairly mild yet technical conditions, thus one may skip this section is such details are not of interest. Ultimately, we will see we only need to assume that the mechanism \(\cM\) is measurable in order to use the Stein-Haff identity.

In this section we will not require the function \(g\) to be defined for singular matrices, as they will not arise in our application of the identity. For the remainder of this section let \(p_{\cW}\) be the density of the Wishart distribution \(\cW_d(D;V)\) for some non-singular \(V \in \PDCone\).
That is, for all \(\Sigma \in \PSDCone\) define
\begin{equation*}
    p_{\cW}(\Sigma) \coloneqq  \frac{1}{2^{Dd/2} \det(V)^{D/2} \Gamma_d(D/2)} \cdot \det(\Sigma)^{(D - d -1)/2} \exp\paren[\Big]{-\frac{1}{2}\Tr(V^{-1}\Sigma)} .
\end{equation*}

The conditions in \citet[Theorem~2.1]{Haff79a} state that 
\begin{enumerate}[(i)]
    \item 
    \label{item:haff_cond1}
    For any \emph{strictly positive} numbers \(\rho_1, \rho_2\), the function $g \cdot p_{\cW}$ should be continuously differentiable (or at least Lipchitz continuous) over the set
    \begin{equation*}
        B(\rho_1, \rho_2) \coloneqq \setst{A \in \PSDCone}{ \rho_1 < \norm{A}_{F} < \rho_2}.
    \end{equation*}
    Note that the matrices in $B(\rho_1,\rho_2)$ can be singular. 
    
    \item
    \label{item:haff_cond2}
     Define \(B(\rho) \coloneqq \setst{A \in \PSDCone}{ \norm{A}_F = \rho}\). We need \(g\) to not grow too fast at the boundaries of \(\PSDCone\), in the sense that
    \begin{equation*}
        \lim_{\rho \to 0} \frac{\sup_{A \in B(\rho)} \norm{g(A)}_F}{\rho^{1-d D/2}} = 0
        \qquad \text{and}\qquad 
        \lim_{\rho \to +\infty} \frac{\sup_{A \in B(\rho)} \norm{g(A)}_F}{\rho^{1-d D/2} \exp\paren{\gamma \rho}} = 0~\text{for every}~\gamma > 0.
    \end{equation*}
\end{enumerate}

Let us now verify that the function \(g \colon \PDCone \to \PSDCone\) given\footnote{One could define the function \(g\) over \(\PSDCone\), but for \(\Sigma\) singular the distribution \(\cN(0, \Sigma)\) does not have a density over \(\bR^d\). In this case, reasoning about differentiability of \(g\) with respect to \(\Sigma\) would be more challenging.} by
\begin{equation}
    \label{eq:def_g_formal}
    g(\Sigma) \coloneqq \expectg{\cM(X)}{\Sigma} = \int_{\bR^{d \times n}} \cM(X) \cdot \paren[\Big]{\prod_{i = 1}^n p(x_i \cond \Sigma)} \diff X, \qquad \forall \Sigma \in \PDCone
\end{equation}
satisfy the above conditions, where \(p(\cdot \cond \Sigma)\) is the density of a normal distribution with mean \(0\) and covariance matrix \(\Sigma\), and \(\cM\) is a mechanism that is measurable and such that \(\cM(X) \in \WellCond\) (we can assume this last condition holds since, if it does not, we may project onto \(\WellCond\) as argued in the proof of \Theorem{Main}). Recall that  \(\WellCond = \setst{A \in \PSDCone}{0.09 I \preceq A \preceq 10 I}\). In this case, \(g(A) \in \cW\) for any~\(A \in \PSDCone\). Therefore, condition~\eqref{item:haff_cond2} is easily satisfied since \(\sup_{A \in \PSDCone} \norm{g(A)}_F\) is bounded by \(10 \sqrt{d}\). 

Condition~\eqref{item:haff_cond1} is used to guarantee that we can apply Stokes' theorem (or a special case of it, the Gauss divergence theorem) to the function \(g \cdot p_{\cW}\) over \(\PSDCone\) (or, actually, on \(B(\rho_1, \rho_2)\) and then take the limits with \(\rho_1\) vanishing and \(\rho_2\) tending to infinity) and requires care to verify. Note that \(B(\rho_1, \rho_2)\) is not an open set (it is an open ball intersected with \(\PSDCone\), and there may be singular matrices in this ball). Thus, we need the existence of a continuously differentiable extension of \(g \cdot p_{\cW}\) over an open set \(D \supseteq B(\rho_1, \rho_2)\) with \(D \subseteq \Sym{d}\). This is often not immediate since \(g\) could be defined only over the set of positive definite matrices (e.g., \(g(\Sigma) = \Sigma^{-1}\)) and its value (or even only the value of its derivatives) would tend to infinity for any sequence approaching the boundary of \(\PSDCone\). Luckily, \(p_{\cW}\) decreases fast enough at the boundary of \(\PSDCone\) that we can easily extend \(g \cdot p_{\cW}\) over \emph{the entire set of symmetric matrices} by setting its value to zero.

Intuitively, as \(\lambda_{\min}(\Sigma_k)\) goes to 0, we have that \(p_{\cW}(\Sigma_k)\) goes to 0 as fast as \(\lambda_{\min}(\Sigma_k)^{D - d - 1}\), while the derivative of \(g(\Sigma)\) will tend to infinity at a speed similar to \(\lambda_{\min}(\Sigma_k)^{-1}\).  We formally show the details of this continuously differentiable extension in the next theorem. If the reader is convinced by the intuitive argument just given, then the proof of the next theorem may be skipped.

\begin{theorem}
    Let \(g\) be defined as in~\eqref{eq:def_g_formal} and define \(F \colon \Sym{d} \to \PSDCone\) by
    \begin{equation*}
      F(\Sigma) \coloneqq
      \begin{cases}
        g(\Sigma) \cdot p_{\cW}(\Sigma), &\text{if}~\Sigma \succ 0,
        \\
        0 & \text{otherwise}.
      \end{cases}  
    \end{equation*}
    Then \(F\) is continuously differentiable. 
\end{theorem}
\begin{proof}
    Let \(i,j \in [d]\). Since \(\cM\) is bounded (that is, always in \(\cW\)) and $p_{\cW}(\Sigma)=0$ for $\Sigma \in \PSDCone \setminus \PDCone$, one can easily verify that \(F\) is continuous. Let us show that \(\partial_{\Sigma_{ij}} F(\Sigma)\) exists and is continuous over \(\Sym{d}\). Over \(\Sym{d}\setminus \PSDCone\) we clearly have \(\partial_{\Sigma_{ij}} F(\Sigma) = 0\). Let us first derive the derivatives at \(\Sigma \in \PDCone\), and then show that the limit at the boundary is 0.
    
    First note that
    \begin{equation}
        \label{eq:partial_F}
        \partial_{\Sigma_{ij}}F(\Sigma)
        = g(\Sigma) \cdot \partial_{\Sigma_{ij}}( p_{\cW}(\Sigma)) + 
        \partial_{\Sigma_{ij}}( g(\Sigma)) p_{\cW}(\Sigma).
    \end{equation}
    To compute $\partial_{\Sigma_{ij}}( g(\Sigma)) = \partial_{\Sigma_{ij}} \ESigma{\cM(X)}$, we will exchange the order of the differential and integration (the expectation). To do so, we may the Leibniz integral rule (see, e.g., ~\citet[Theorem~2.27]{Folland99a}), which requires us to check that the partial derivative of the integrand in the right-hand side~\eqref{eq:def_g_formal} is bounded in absolute value by an integrable function for any \(\Sigma\) on an open ball contained in \(\PDCone\). That is, we will show
    \begin{claimeq}
        \label{eq:claim_dominance}
        given \(\Sigmabar \in \PDCone\) and \(\eps > 0\) small enough,  there is a function \(H(X)\) and a constant \(C > 0\) (that may depend on \(\Sigmabar\) and \(d\)) such that, for any \(\Sigma \in \PDCone\) with \(\norm{\Sigmabar - \Sigma}_F < \eps\), we have
    \end{claimeq}
    \begin{equation*}
        \partial_{\Sigma_{ij}} \paren[\Big]{ \cM(X) \cdot \paren[\Big]{\prod_{k = 1}^n p(x_k \cond \Sigma)}} \leq C H(X) \qquad \text{and} \qquad \int H(X) \diff X < \infty.
    \end{equation*}
    Fix \(\Sigmabar \in \PDCone\) and let \(\eps > 0\) be small enough such that \(\mB_{\eps} \coloneqq \{\Sigma \in \Sym{d}\colon \norm{\Sigma - \Sigmabar}_F^2 < \eps\} \subseteq \PDCone\). Let \(\Sigma \in \mB_{\eps}\). For any \(x \in \bR^d\) define \(s(x)_{ij} \coloneqq \partial_{\Sigma_{ij}}(\ln p(x \cond \Sigma)) = \partial_{\Sigma_{ij}}(p(x \cond \Sigma)) / p(x \cond \Sigma_{ij})\).  Then,
    \begin{align*}
        \partial_{\Sigma_{ij}} \paren[\Big]{ \cM(X) \cdot \paren[\Big]{\prod_{k = 1}^n p(x_k \cond \Sigma)}}
        &= 
         \cM(X) \cdot \partial_{\Sigma_{ij}}\paren[\Big]{\prod_{k = 1}^n p(x_k \cond \Sigma)} 
        \\
        &=
         \cM(X) \cdot \sum_{k = 1}^n \partial_{\Sigma_{ij}} p(x_i \cond \Sigma) \paren[\Big]{\prod_{r \neq i} p(x_r \cond \Sigma)}
        \\
        &=
        \cM(X) \cdot \paren[\Big]{\sum_{k = 1}^n s(x_k)_{ij} }\paren[\Big]{\prod_{k = 1}^n p(x_k \cond \Sigma)}.
    \end{align*}
    To conclude the proof of~\eqref{eq:claim_dominance}, it suffices to show that (each entry of) the above expression is upper-bounded by an integrable function of~\(x_1, \dotsc, x_n\) with respect to the Lebesgue measure. For that, it suffices to show that \(\ESigma{\cM(X) s(x_k)_{ij}}\) is finite with \(x_1, \dotsc,x_k \sim \cN(0, \Sigma)\).  Notice that any entry of \(\cM(X)s(x_k)_{ij} \) is upper bounded by its maximum eigenvalue. Using the formula for \(s(x)\) (see Proposition~\ref{prop:Score} and Lemma~\ref{lemma:sym_gradient_vech_gradient}) and the fact that \(\cM(X) \preceq 10I\) we have 
    \begin{align*}
        \lambda_{\max}(\cM(X) \cdot s(x_k)_{ij})
        &\lesssim  \abs{s(x_k)_{ij}}
        \lesssim \lambda_{\max}\paren[\Big]{\Sigma^{-1} x_k x_k^\T \Sigma^{-1} - \Sigma^{-1}}\\
        &=\lambda_{\max}(\Sigma^{-1})( \lambda_{\max}((\Sigma^{-1/2}x_k) (\Sigma^{-1/2}x_k)^{\T}) - 1) 
        \\
        &= \lambda_{\max}(\Sigma^{-1})(\lVert \Sigma^{-1/2}x_k\rVert_2^2 -1),
    \end{align*} 
    where, as usual, the \(\lesssim\) omits global constants.
    Since \(x_k \sim \cN(0, \Sigma)\), we have \(\Sigma^{-1/2} x_k \sim \cN(0,I)\) and, thus,  \(\expectsmall{\norm{\Sigma^{-1/2} x_k}_2^2} = d\). Moreover, since \(\Sigma \in \mB_\eps\), we have the (loose) bound \(\lambda_{\max}(\Sigma^{-1}) \leq d\eps \lambda_{\max}(\Sigmabar^{-1})\). This finishes the proof of~\eqref{eq:claim_dominance}.
    
    Thus, applying the Leibniz integral rule, we can exchange differentiation and integration, obtaining
    \begin{equation}
        \label{eq:formula_partials}
        \partial_{\Sigma_{ij}} g(\Sigma) 
        =  \ESigmaBig{\cM(X) \cdot \sum_{k = 1}^n s(x_k)_{ij}}
    \end{equation}
    Additionally, the dominated convergence theorem (which is applicable due to~\eqref{eq:claim_dominance}) also implies that the above function is continuous on \(\PDCone\).

    Next, let us derive an expression for \(\partial_{\Sigma_{ij}} p_{\cW}(\Sigma)\).
    Recalling the formula for the gradient\footnote{See Section~\ref{app:SymmetricGradient} for a discussion on gradients for functions of symmetric matrices and Lemma~\ref{lemma:sym_gradient_vech_gradient} for a result showing that the partial derivatives are scaled entries of the gradient.} of \(p_{\cW}\) in~\eqref{eq:DerivativeOfDensity} is given by
    \begin{equation}
        \label{eq:gradient_wishart}
        \nabla p_{\cW}(\Sigma)
        = \frac{1}{2}(V^{-1} - (D - d - 1)\Sigma^{-1})p_{\cW}(\Sigma).
    \end{equation}
    Therefore, \(\partial_{\Sigma_{ij}} p_{\cW}\) is continuous on \(\PDCone\). This, together with the continuity of \(\partial_{\Sigma_{ij}}g\) and~\eqref{eq:partial_F}, implies that \(\partial_{\Sigma_{ij}} F\) is continuous on \(\PDCone\).

    It only remains to show continuity of \(\partial_{\Sigma_{ij}}F(\Sigma)\) at the boundary \(\bd(\PSDCone) \coloneqq \PSDCone \setminus \PDCone\).
    Let \((\Sigma_k)_{k = 1}^{+\infty}\) be a convergent sequence in \(\PDCone\) such that \(\lim_{k \to \infty} \Sigma_k = \Sigmabar \in \bd(\PSDCone)\). We shall show that
    \begin{equation}
        \label{eq:lim_derivatives_claim}
        \lim_{k \to \infty} \partial_{\Sigma_{ij}}( g(\Sigma_k)) p_{\cW}(\Sigma_k) 
        =\lim_{k \to \infty} g(\Sigma_k) \partial_{\Sigma_{ij}}( p_{\cW}(\Sigma_k)) 
        = 0.
    \end{equation}
    This, together with the expression in~\eqref{eq:partial_F} implies that \(\partial_{\Sigma_{ij}} F(\Sigmabar) = 0\), as desired. 
    For the remainder of the proof, we shall write \(a \lesssim b \) if \(a \leq C \cdot b\) where \(C\) is a constant that is independent of the sequence \(\Sigma_k\) (the constant $C$ \emph{may depend on parameters such as}~\(d\) and \(n\)). Then, for any index \(k\), 
    \begin{align*}
        &\norm{\partial_{\Sigma_{ij}}( g(\Sigma_k)) p_{\cW}(\Sigma_k)}_F
        \\
        &\leq \sum_{r = 1}^n \ESigma{\norm{\cM(X) s(x_r)_{ij}}_F^2}^{1/2} p_{\cW}(\Sigma_k) &\text{(By~\eqref{eq:formula_partials} and Jensen's ineq.)}
        \\
        &\lesssim \sum_{r = 1}^n \ESigma{s(x_r)_{ij}^2}^{1/2} p_{\cW}(\Sigma_k)
        & \text{($\cM(X)$ is bounded)}
        \\
        &\lesssim \ESigma{s(x_1)_{ij}^2}^{1/2} p_{\cW}(\Sigma_k)
        & \text{($x_1, \dotsc, x_n$ are i.i.d.)}
        \\
        &= \cI(\Sigma_k)_{ij, ij}^{1/2} \cdot p_{\cW}(\Sigma_k)
        & \text{(By the def.~of Fisher info. from~\eqref{eq:def_fisher})}
        \\
        &\lesssim \lambda_{\max}(\cI(\Sigma_k))^{1/2} \cdot p_{\cW}(\Sigma_k)
        \\
        &\leq \lambda_{\min}(\Sigma_k)^{-1} \cdot p_{\cW}(\Sigma_k)
        &\text{(Lemma~\ref{lemma:fisher_information}).}
    \end{align*}
    For any matrix \(A \in \Sym{d}\), let \(\lambda_1(A), \cdots, \lambda_d(A)\) be the eigenvalues of \(A\) in non-increasing order. Since \(\Sigma_k \to \Sigmabar\) as \(k \to \infty\), we have \(\lambda_r(\Sigma_k) \to \lambda_r(\Sigmabar)\) for every \(r \in [d]\). In particular, we have the limit \(\lim_{k \to \infty} \prod_{r< n} \lambda_r(\Sigma_k) = \prod_{r< n} \lambda_r(\Sigmabar)\). Therefore, if \(D > d - 2\) 
    \begin{align*}
        \lambda_{\min}(\Sigma_k)^{-1} \cdot p_{\cW}(\Sigma_k)
        &\lesssim \lambda_{\min}(\Sigma_k)^{-1} \cdot \det(\Sigma_k)^{D - d - 1}
        =   \det(\Sigma_k)^{D - d - 2} \prod_{r< n} \lambda_r(\Sigma_k)
        \\
        &\stackrel{k \to 
        \infty}{\longrightarrow}
        {\underbrace{\det(\Sigmabar)}_{= 0}}^{D - d - 2} \prod_{r< n} \lambda_r(\Sigmabar) = 0.
    \end{align*}

    So far we have analyzed one of the limits in \eqref{eq:lim_derivatives_claim}.
    The other limit can be analyzed by similar arguments,
    also under the assumption that \(D > d-2\).
    Using the formula for the gradient in~\eqref{eq:gradient_wishart}, but omitting details for the sake of conciseness, we obtain
    \begin{align*}
        \norm{g(\Sigma) \cdot \partial_{\Sigma_{ij}}( p_{\cW}(\Sigma))}_F 
        ~\lesssim~ (\norm{V}_F + \norm{\Sigma_k^{-1}}_F) \cdot p_{\cW}(\Sigma_k)
        ~\lesssim~ \frac{p_{\cW}(\Sigma_k)}{\lambda_{\min}(\Sigma)}
        ~\stackrel{k \to 
        \infty}{\longrightarrow}~ 0.
    \end{align*}
    This concludes the proof of~\eqref{eq:lim_derivatives_claim}.
\end{proof}

\section{Omitted proofs}

\subsection{Omitted material from \Section{Overview}}
\AppendixName{OmittedOverview}

\begin{proofof}{Lemma~\ref{lemma:upper_bound_covariance_randomized_sigma}}
    Taking the expectation with respect to \(\Sigma\) on the inequality of Lemma~\ref{lemma:upper_bound_covariance}, then using the tower property of conditional expectation, yields 
    \begin{align*}
        \expect{\sum_{i = 1}^n \ESigma{\cA(x_i, \cM(X))}}
        &~=~
        \sum_{i = 1}^n \expect{\cA(x_i, \cM(X))} \\
        &~\leq~ 2 n \eps  \expect{\frac{\alpha_{\Sigma}}{\lambda_{\min}(\Sigma)}}   + \paren*{10\expect{\frac{1}{\lambda_{\min}(\Sigma)}} + 1} d^{3/2}.
    \end{align*}
    A standard bounds for Wishart matrices (see     Lemma~\ref{lemma:bound_expected_inv_lambda_min})
    is that $\expect{\lambda_{\min}^{-1}(\Sigma)} \leq 6.5$, so the second term on the right-hand side is at most \(66 d^{3/2}\).
    Thus, it remains to show that there is a constant \(C > 0\) such that
    \begin{equation}
        \label{eq:bound_alpha_sigma_lambda_min}
        \expect{\alpha_{\Sigma}/\lambda_{\min}(\Sigma)} \leq 10 \alpha + C 2^{-d}.    
    \end{equation}
    We will do so by separately bounding \(\expect{(\alpha_{\Sigma}/\lambda_{\min}(\Sigma)) \cdot \ones_{\cE}}\) and \(\expect{(\alpha_{\Sigma}/\lambda_{\min}(\Sigma) )\cdot \ones_{\bar{\cE}}}\),
    where $\cE$ is the event $\curly{ \Sigma \in \WellCond}$.
    Under \(\cE\) we have \(\lambda_{\min}(\Sigma) \geq 0.09 \geq 1/12\) and, by the definition of \(\alpha\) (see \eqref{eq:alphadef}), we have \(\expect{\alpha_{\Sigma} \cdot \ones_{\cE}} \leq \expect{\alpha_{\Sigma} \cond \cE} = \alpha\). Therefore, \(\expect{(\alpha_{\Sigma}/\lambda_{\min}(\Sigma)) \cdot \ones_{\cE}} \leq 12 \alpha\). For the other term, first we can use the fact that \(\cM(X) \preceq 10 I\) to get
    \begin{align*}
        \alpha_{\Sigma}
        &= \sqrt{\ESigma{\norm{\cM(X) - \Sigma}_F^2}}
        \leq \sqrt{\ESigma{\norm{\cM(X) - \Sigma}_F^2}}
        \leq \sqrt{\ESigma{2\norm{\cM(X)}_F^2 + 2\norm{\Sigma}_F^2}}
        \\
        &\leq \sqrt{2d(10^2 + \lambda_{\max}(\Sigma)^2)}
        \leq \sqrt{2d}(10 + \lambda_{\max}(\Sigma)).
    \end{align*}
    Define \(\kappa(\Sigma) \coloneqq {\lambda_{\max}(\Sigma)}/{\lambda_{\min}(\Sigma)}\). Then,
    \begin{equation*}
        \expect{\frac{\alpha_{\Sigma}}{\lambda_{\min}(\Sigma)} \cdot \ones_{\bar{\cE}}}
        \leq 10\sqrt{2d}\expect{\lambda_{\min}(\Sigma)^{-1} \cdot \ones_{\bar{\cE}}} + \sqrt{2d} \expect{\kappa(\Sigma) \cdot \ones_{\bar{\cE}}}.
    \end{equation*}
    Let us start by noticing that  Lemmas~\ref{lemma:tail_bound_lambdamax_sigma} and~\ref{lemma:bound_expected_inv_lambda_min} together yield
    \begin{equation}
        \label{eq:bound_prob_bounded_sigma}
        \prob{\bar{\cE}} \leq \prob{\lambda_{\min}(\Sigma) < 0.09} + \prob{\lambda_{\max}(\Sigma) > 10}
        \leq \frac{2}{\sqrt{d}} 2^{-d}.
    \end{equation}
    Using the tail bound~\eqref{eq:tail_bound_lambda_min_sigma} we have
    \begin{align*}
        \expect{\frac{1}{\lambda_{\min}(\Sigma)}\cdot \ones_{\bar{\cE}}}
        &= \int_{0}^{\infty} \prob{\frac{1}{\lambda_{\min}(\Sigma)} \geq t \;\wedge\; \bar{\cE} } 
        \leq 4e \cdot \prob{\bar{\cE}} + \int_{4e}^{\infty} \prob{\lambda_{\min}(\Sigma)^{-1} > t} \diff t 
        \\
        &\leq 
        4e \cdot \prob{\bar{\cE}} + \int_{4e}^{\infty}  \frac{1}{\sqrt{2\pi d}} \paren*{2e}^{d + 1} \frac{1}{t^{d+1}} \diff t 
        = 4e \cdot \prob{\bar{\cE}} + \frac{1}{\sqrt{2\pi d}}  \frac{1}{d}\frac{\paren*{2e}^{d + 1}}{(4e)^{d}}
        \\
        &\lesssim \frac{1}{\sqrt{d}} 2^{-d} + \frac{1}{\sqrt{d}} 2^{-d}
        \leq \frac{2}{\sqrt{d}} 2^{-d},
    \end{align*}
    where in the last step we used \(d \geq 2e \).
    Similarly, using the tail bound in Lemma~\ref{lemma:tail_cond_number_sigma},
    \begin{align*}
        \expect{\kappa(\Sigma) \cdot \ones_{\bar{\cE}}}
        &\leq 10^4 \prob{\bar{\cE}} + \int_{10^4}^{\infty} \frac{(13)^{d+1}}{\sqrt{2\pi}} \cdot \frac{1}{t^{(d+1)/2}} \diff t
        \\
        &\lesssim \frac{1}{\sqrt{d}} 2^{-d} + \int_{10^4}^{\infty} (13)^{d+1} \cdot \frac{1}{t^{(d+1)/2}} \diff t
        \\
        &= \frac{1}{\sqrt{d}} 2^{-d} + 13 \frac{2}{d+1} \frac{13^{d}}{(10^4)^{(d-1)/2}}
        \\
        &\lesssim  \frac{1}{\sqrt{d}} 2^{-d} + \frac{1}{d} \paren*{\frac{13}{10^2}}^{d}
        \\
        &\leq \frac{1}{\sqrt{d}} 2^{-d} + \frac{1}{\sqrt{d}} 2^{-d} .
    \end{align*}
    Putting everything together yields
    \begin{equation*}
        \expect{\frac{\alpha_{\Sigma}}{\lambda_{\min}(\Sigma)}\cdot \ones_{\bar{\cE}}}
        \lesssim \sqrt{d} \cdot \frac{1}{\sqrt{d}} 2^{-d} \leq 2^{-d},
    \end{equation*}
    which finishes the proof of \eqref{eq:bound_alpha_sigma_lambda_min}.
\end{proofof}

\subsection{Omitted material from \Section{LB_haff}}
\AppendixName{OmittedLB}

\begin{lemma}
    \label{lemma:expect_sigma_squared}
    For \(\alpha\) and \(\alpha_{\Sigma}\) defined as in~\eqref{eq:alphadef} we have
    \begin{equation*}
        \expect{\alpha_{\Sigma}^2}
        = \expect{\norm{\cM(X) - \Sigma}_{F}^2}
        \leq \alpha^2 + 600 d \cdot \frac{2^{-d}}{\sqrt{d}}.
    \end{equation*}
    In particular, if \(d \geq 19\), then $\expect{\alpha_{\Sigma}^2} \leq \alpha^2 + d/200$.
\end{lemma}
\begin{proof}
    Define the event \(\cE \coloneqq \curly{\Sigma \in \WellCond}\). Then, by the definition of \(\alpha\),
    \begin{equation*}
        \expect{\alpha_{\Sigma}^2} 
        = \expect{\alpha_{\Sigma}^2 \ones_{\cE}}
        + \expect{\alpha_{\Sigma}^2 \ones_{\bar{\cE}}}
        \leq \alpha^2 + \expect{\alpha_{\Sigma}^2 \ones_{\bar{\cE}}}.
    \end{equation*}
    Moreover, similarly to what we did in the proof of Lemma~\ref{lemma:upper_bound_covariance_randomized_sigma}, we have
    \begin{align*}
        \expect{\alpha_{\Sigma}^2 \ones_{\bar{\cE}}}
        &\leq \expect{2d(10^2 + \lambda_{\max}(\Sigma)^2)\ones_{\bar{\cE}}} = 200d \prob{\bar{\cE}}
        + d\expect{\lambda_{\max}(\Sigma)^2 \ones_{\bar{\cE}}} 
        \\
        &\stackrel{\eqref{eq:bound_prob_bounded_sigma}}{\leq}
        200d \frac{2}{\sqrt{d}} 2^{-d}
        + d\expect{\lambda_{\max}(\Sigma)^2 \ones_{\bar{\cE}}} 
    \end{align*}
    Let us now give an \(O(2^{-d}/\sqrt{d})\) upper bound the latter term. We have,
    \begin{align*}
        \expect{\lambda_{\max}(\Sigma)^2 \ones_{\bar{\cE}}}
        &\leq 80 \cdot \prob{\bar{\cE}} + \int_{80}^{\infty} \prob{\lambda_{\max}(\Sigma)^2 \geq t} \diff t
        \\
        &= 80 \cdot \prob{\bar{\cE}} + \int_{8}^{\infty} \prob{\lambda_{\max}(\Sigma) \geq \sqrt{72 + s}} \diff s
        \\
        &\leq    80 \cdot \prob{\bar{\cE}} + \int_{8}^{\infty} \prob{\lambda_{\max}(\Sigma) \geq \underbrace{\sqrt{72/2}}_{= 6} + \sqrt{s/2}} \diff s
        \\
        &\leq    80 \cdot \prob{\bar{\cE}} + \int_{8}^{\infty} \exp\paren*{\frac{-d \sqrt{s}}{2 \sqrt{2}}} \diff s
        & \text{(By Lemma~\ref{lemma:tail_bound_lambdamax_sigma})}
        \\
        &\leq    80 \cdot \prob{\bar{\cE}} + \frac{4 \sqrt{2} (d \sqrt{8} + 2 \sqrt{2})}{d^2} \exp\paren*{\frac{-d \sqrt{8}}{2 \sqrt{2}}} 
        \\
        &\leq    80 \cdot \prob{\bar{\cE}} + \frac{16(d+1)}{d^2} e^{-d}
        \\
        &\leq  160 \cdot \frac{1}{\sqrt{d}} 2^{-d}+ \frac{160}{9 d} e^{-d}
        & (d \geq 9)
        \\
        &\leq  160 \cdot \frac{10}{9\sqrt{d}} 2^{-d}.
    \end{align*}
    Thus,
    \begin{equation*}
        \expect{\alpha_{\Sigma}^2 \ones_{\bar{\cE}}} \leq d \cdot \frac{2^{-d}}{\sqrt{d}}\paren*{400 + \frac{1600}{9}} \leq 600 d \cdot \frac{2^{-d}}{\sqrt{d}} 
    \end{equation*}
    In particular, one may verify that for \(d \geq 19\) we have \(600 \sqrt{d} 2^{-d} \leq 1/200\)
\end{proof}




\section{Score and Fisher Information of Gaussian with Unknown Covariance}
\AppendixName{ScoreFisherAppendix}

In this section we shall rigorously derive the formulas for the score function and the Fisher information matrix, with respect to the parameter \(\Sigma \in \PSDCone\), for the Gaussian density
\begin{equation*}
    p(x \cond \Sigma)
    = \frac{1}{(2 \pi)^{d/2} \det(\Sigma)^{1/2}} \exp\paren[\Big]{-\frac{1}{2}x^{\transp} \Sigma^{-1} x}, \qquad \forall x \in \bR^d.
\end{equation*}
Since the function is meant to be evaluated only for symmetric, positive definite matrices \(\Sigma\), one should take into account the symmetry of \(\Sigma\) when manipulating the derivatives and gradients of \(p(\cdot \cond \Sigma)\) with respect to \(\Sigma\) (see, e.g., the discussions in~\citealp{SrinivasanP23a}).

If fact, these results are already known (see \citealp[\S 5]{MagnusN80a} or \citealp{Barfoot20a}), but they are not always derived with the adequate amount of rigor. Furthermore, one may find a variety of different versions of the Fisher information, both depending on the parameterization and how one takes symmetry in account during differentiation. It shall be important for us to rigorously derive the score and Fisher information to then connect these results to the Stein-Haff identity in Section~\ref{sec:LB_haff}. Moreover, we shall derive a formula for the Fisher information directly from its definition as the covariance matrix of the score, without making use of its connection to the second derivative of the score (which needs care when taking symmetry into account).  

We shall first discuss how to properly take the symmetry in differentiation. Then we shall derive the formulas of the score function (including differentiating between ``matrix score'' and the classical score function) and the Fisher information matrix of the Gaussian density.  

\subsection{Gradients of Functions of Symmetric Matrices}
\AppendixName{SymmetricGradient}

Let \(F \colon \Sym{d} \to \Reals\) be a real valued function over symmetric matrices. (Later we will focus on the function \(\Sigma \in \PSDCone \mapsto \ln p(x \cond \Sigma)\).)
Due to symmetry,  we may actually consider \(F\) to be a function of the lower triangular portion of the matrix, that is, restricted to \(\Sym{d}\), the function \(F\) is a real function over a space isomorphic to \(\bR^{d(d+1)/2}\). Formalizing this turns out to be crucial to correctly define and compute the score function and Fisher information matrices of a Gaussian with unknown covariance.

More formally, define \(\binom{[d]}{k} \coloneqq \setst{S \subseteq [d]}{\card{S} = k}\). For any matrix \(A \in \bR^{d \times d}\) (not necessarily symmetric), we shall denote the vector containing the entries in the lower triangular portion of \(A\) by \(\vech(\Sigma)\) which we paraneterize it in the natural way by \(\binom{[d]}{2} \cup \binom{[d]}{1}\). 
Formally, for any matrix \(A \in \bR^{d \times d}\) we define the vector \(\vech(A)\colon \binom{[d]}{2} \cup \binom{[d]}{1} \to \Reals\) by
\begin{equation*}
    \vech(A)_{\{i,j\}} \coloneqq A_{ij} \qquad \forall i,j \in [d]~\text{with}~i \geq j.
\end{equation*}
Analogously, we define the \(d^2\) dimensional vector\footnote{For the sake of preciseness, we shall for this section differentiate the space \(\bR^{d \times d}\) of \(d \times d\) matrices and the space \(\bR^{[d] \times [d]}\) of vectors indexed by ordered pairs of elements in \([d]\) so that we can use vector operations in elements of \(\bR^{[d] \times [d]}\) without requiring us to overload notation in the space of matrices.} \(\vec(A) \in \bR^{[d] \times [d]}\) by
\begin{equation*}
    \vec(A)_{(i,j)} \coloneqq A_{ij}, \qquad \forall i,j \in [d].
\end{equation*}

The idea of the duplication matrix \citep{MagnusN80a} will be useful in navigating these different spaces. The \emph{Duplication matrix} is the matrix \(D \colon ([d] \times [d]) \times (\binom{[d]}{2} \cup \binom{[d]}{1}) \to \{0,1\}\) given by
\begin{equation*}
     D((r,s), \{i,j\}) = \boole{\{r,s\} = \{i,j\}}, \qquad \forall i,j,r,s \in [d].
\end{equation*}
In words, for any matrix \(A \in \bR^{d \times d}\), the vector \(D \vech(A)\) is the \(d^2\) dimensional vector comprised of the lower triangle entries of \(A\), with the off-diagonal entries duplicated. If \(A\) is symmetric then \(D \vech(A) = \vec(A)\). Following~\citet{SrinivasanP23a}, define the \emph{symmetric gradient of \(F\)} by\footnote{This definition requires \(F\) to be well-defined and differentiable over non-symmetric matrices. \citet{SrinivasanP23a} also show how to write this gradient when we have a function \(F\) that is \emph{only} defined over symmetric matrices and cannot be extended to \(\bR^{d \times d}\).}
\begin{equation}
    \label{eq:matrix_gradient_identity}
    \nabla_{\sym} F(A) = \frac{1}{2}(\nabla F(A) + \nabla F(A)^\transp),
\end{equation}
where \(\nabla F(A) \in \bR^{d \times d}\) is the traditional gradient of \(F\), that is, \((\nabla F(A))_{ij}\) is the derivative of \(F\) evaluated at \(A\) \emph{not taking into account symmetry}. This definition is far from arbitrary:~\citet{SrinivasanP23a} show that \(\nabla_{\sym} F(\Sigma)\) is the gradient that satisfies the definition of \emph{Frechét derivative} over \(\Sym{d}\), the space of symmetric \(d \times d\) matrices. This identity shall be useful later to compute the score function of a Gaussian with respect to the covariance matrix using traditional matrix calculus.

Finally, the following lemma connects the symmetric gradient of \(F\) and the gradient of the function \(f\) given by \(f(\vech(A)) \coloneqq F(A)\) for any \(A \in \Sym{d}\). Namely, it shows that \(\nabla_{\sym} F\) is defined in such a way such that \(\iprod{\nabla_{\sym} F(A)}{B}\) agrees with \(\iprodt{\nabla f(\vech(A))}{\vech(B)}\).

\begin{lemma}
    \label{lemma:sym_gradient_vech_gradient}
    Let \(F \colon \Sym{d} \to \bR\) and \(f\colon \bR^{\binom{[d]}{2} \cup \binom{[d]}{1}} \to \Reals\) be such that \(F(A) = f(\vech(A))\) for all \(A \in \Sym{d}\). Then, for all \(i,j \in [d]\) and any \(A \in \Sym{d}\), we have
    \begin{equation*}
        \nabla_{\sym}F(A)_{ij} = \frac{(1 + \boole{i \neq j})}{2}\nabla f(\vech(A))_{\{i,j\}}.
    \end{equation*}
    In particular, for any symmetric matrix \(B\) we have \(\iprod{B}{\nabla_{\sym} F(A)} = \iprodt{\vech(B)}{\nabla f(\vech(A))}\).
\end{lemma}
\begin{proof}
    This is a combination of~\citet[Theorem~3.8]{SrinivasanP23a} and the formula for the pseudo-inverse of \(D\) given by~\citet[Lemma 3.6.iv]{MagnusN80a}.
\end{proof}

\subsection{Score Function Derivation}

Let us start by discussing the definition and derivation of the score function of the Gaussian distribution \(\cN(0, \Sigma)\) with respect to the covariance matrix \(\Sigma \in \PSDCone\). As discussed in the previous section, to properly take into account the symmetry of the covariance matrix \(\Sigma\), we should look at the function \(f_x \colon \bR^{\binom{[d]}{2} \cup \binom{[d]}{1}} \to \Reals\) defined by
\begin{equation}
    \label{eq:vech_log_prob}
    f_x(\vech(\Sigma)) \coloneqq \ln p(x \cond \Sigma), \qquad \forall \Sigma \in \bR^{d \times d}, \qquad \forall x \in \bR^d,
\end{equation}
where we let \(f_x(\vech(\Sigma))\) evaluate to \(+\infty\) whenever \(p(x \cond \Sigma)\) is not well-defined. Then, the \emph{score function} of \(p(x \cond \Sigma)\) with respect to \(\Sigma\) is \(s(x) \coloneqq (\nabla f_x) (\vech(\Sigma))\). Here we use extra parentheses to make it clear that the latter expression if the gradient of \(f_x\) evaluated at \(\vech(\Sigma)\), not the gradient of \(f \circ \vech\) evaluated at \(\Sigma\).  One issue that could slightly complicate the computation of the score is that we know the formulas for \(p(x \cond \Sigma)\) is matrix notation, and differentiating with respect to each entry individually could be cumbersome. Moreover, as demonstrated by~\citet{SrinivasanP23a}, for any \(i,j \in [d]\), we probably have \(s(x)_{\{i,j\}} = (\nabla f_x) (\vech(\Sigma))_{\{i,j\}} \neq \nabla F_x(\Sigma)_{ij}\) where
\begin{equation*}
    F_x(\Sigma) \coloneqq \ln p(x \cond \Sigma)  \qquad \forall \Sigma \in \bR^{d \times d}, \qquad \forall x \in \bR^d
\end{equation*}
and the gradient \(\nabla F_x(\Sigma)\) does not take into account symmetry (that is, it is the gradient of \(F_x\) as a function from \(\bR^{d \times d}\) to \(\bR\)). In the literature symmetry has been taken into account in different ways that may disagree with each other, and we refer the interested reader to the discussion in~\citet{SrinivasanP23a}. 

Nonetheless, the identity in~\eqref{eq:matrix_gradient_identity} allows us to compute the symmetric gradient using matrix calculus rules, and Lemma~\ref{lemma:sym_gradient_vech_gradient} allows us to compute the actual score function from the symmetric gradient. In the next proposition, we show that the symmetric gradient and the classical gradient (that does not take into account symmetry) luckily agree in your case and have a simple formula.

\begin{proposition}
    \label{prop:Score}
    Let \(\Sigma \in \PSDCone\) be positive definite and \(x \in \bR^d\). Then,
            \begin{equation*}
            \nabla_{\sym} F_x (\Sigma)
            =
            \nabla F_x(\Sigma)
            = \frac{1}{2} (\Sigma^{-1} x x^{\transp} \Sigma^{-1} - \Sigma^{-1}).
        \end{equation*}
\end{proposition}
\begin{proof}
    First, note that 
    \begin{equation*}
        \ln p(x \cond \Sigma) = - \frac{1}{2} x^{\transp} \Sigma^{-1} x - \frac{1}{2} \ln \det \Sigma - \frac{d}{2} \ln 2\pi.
    \end{equation*}
    Therefore, to compute the score function \(\nabla_{\Sigma} \ln p(x \cond \Sigma)\) it suffices to compute the gradient of each of the terms above individually. The last term is constant with respect to \(\Sigma\), so its derivative is zero. Moreover, can compute the gradients (with respect to \(\Sigma\)) of each term above. Namely,by equations 57 and 61 of \citet{MatrixCookBook} we have
    \begin{equation*}
        \nabla \paren[\big]{ x^{\transp} \Sigma^{-1} x} = - \Sigma^{-1} x x^{\transp} \Sigma^{-1}
        \qquad \text{and} \qquad
        \nabla( \ln \det \Sigma) = \Sigma^{-1}.
    \end{equation*}
    Putting everything together yields the formula we desired. Moreover, since the gradient is already symmetric, the identity~\eqref{eq:matrix_gradient_identity} yields the first equation in the claim.
\end{proof}

\subsection{Fisher Information Derivation}

In the previous section we have obtained a formula for the ``matrix score function'' which, thanks to Lemma~\ref{lemma:sym_gradient_vech_gradient}, allows us to obtain a formula for the actual score function. In this section we shall derive a formula for the Fisher information matrix. As mentioned at the beginning of this section, although the formula for the Fisher information is known in the literature, we did not find a direct proof from the definition of Fisher information. Moreover, since the use of derivatives and gradients of symmetric gradient has not always been rigorous in previous work, we include a rigorous derivation of the formula.

The \emph{Fisher information matrix} of the Gaussian distribution \(\cN(0, \Sigma)\) with respect to \(\Sigma \in \PSDCone\) is the \((d(d+1)/2) \times (d(d+1)/2)\) dimensional matrix given by
\begin{equation}
    \label{eq:def_fisher}
    \cI \equiv \cI(\Sigma) \coloneqq \expect{(\nabla f_x)(\vech(\Sigma)) (\nabla f_{x})(\vech(\Sigma))^{\transp}} \quad \text{where}~x\sim \cN(0, \Sigma),
\end{equation}
where \(f_x\) is the log-density defined in~\eqref{eq:vech_log_prob}. 
To compute a formula for \(\cI\), we will need the following corollary of Isserli's Formula~\citep{Isserlis18a}.

\begin{lemma}
    \label{lemma:fourth_moment}
    For any \(B \in \bR^{d \times d}\) and \(x \sim \cN(0, \Sigma)\) we have
    \begin{equation*}
        \expect{x x^{\transp} B x x^{\transp}} = \Sigma (B + B^{\transp})\Sigma + \Sigma \Tr(B \Sigma).
    \end{equation*}
\end{lemma}
\begin{proof}
    By Isserli's Theorem~\citep{Isserlis18a}, for any \(i,j,r,s \in [d]\) we have
    \begin{equation*}
        \expect{x_i x_r x_s x_j} = 
        \Sigma_{ir} \Sigma_{sj} + \Sigma_{is} \Sigma_{rj} + \Sigma_{ij} \Sigma_{rs}.
    \end{equation*}
    Therefore, for any \(i,j \in [d]\), we have
    \begin{align*}
        \expect{x x^\T B x x^T }_{ij}
        &= 
        \sum_{r,s \in [d]} 
        \expect{x_i x_r B_{r,s} x_s x_j}
        = \sum_{r,s \in [d]} B_{rs} ( \Sigma_{ir} \Sigma_{sj} + \Sigma_{is} \Sigma_{rj} + \Sigma_{ij} \Sigma_{rs}) 
        \\
        &= \sum_{r,s \in [d]} B_{rs} (\Sigma_{ir} \Sigma_{sj} + \Sigma_{is} \Sigma_{rj}) 
        + \Sigma_{ij} \underbrace{\sum_{r,s \in [d]} B_{rs} \Sigma_{rs}}_{= \Tr(B \Sigma^\T) =  \Tr(B \Sigma)}. 
    \end{align*}
    Finally, if \(e_i \in \{0,1\}^d\) denotes the indicator vector given by \(e_i(j) \coloneqq \boole{i = j}\) for any \(j \in [d]\), then
    \begin{equation*}
        \sum_{r,s \in [d]} B_{rs} (\Sigma_{ir} \Sigma_{sj} + \Sigma_{is} \Sigma_{rj}) 
        = \sum_{r,s \in [d]}  \Sigma_{is}(B_{sr} + B_{rs}  ) \Sigma_{rj} 
        = e_i^{\T}\Sigma (B + B^{\T}) \Sigma e_j^{\T}, 
    \end{equation*}
    which concludes the proof of the desired identity.
\end{proof}

Finally, let us derive the formula for the Fisher information matrix \(\cI\). We note that the next theorem can be found in the literature (e.g.,~\citealt[Lemma~5.2]{MagnusN80a}). Yet, most of the derivations rely on looking at the second derivate of the log-density, which requires care to properly take symmetry into account. Ours is a direct derivation from the definition of Fisher information. In the next proposition, \(A \otimes B\) denotes the \emph{Kronecker product} of the matrices \(A\) and \(B\) with appropriate dimensions.

\begin{proposition}
    \label{prop:fisher_info_formula}
    Let \(\Sigma \in \PSDCone\) be non-singular. Then,
    \begin{equation*}
        \cI(\Sigma) = \frac{1}{2} D^{\transp} (\Sigma^{-1} \otimes \Sigma^{-1}) D
    \end{equation*}
\end{proposition}
\begin{proof}
    Let \(x \sim \cN(0, \Sigma)\). 
    Define the (matrix) score function \(S(x) \coloneqq \nabla_{\Sigma}( \ln p(x \cond \Sigma))\), by Proposition~\ref{prop:Score}, we have \(S(x) = \frac{1}{2} (\Sigma^{-1} x x^{\transp} \Sigma^{-1} - \Sigma^{-1})\). Note, however, that the definition of Fisher information depends on the vector score \(s(x) \coloneqq \nabla f_x(\Sigma)\). By Lemma~\ref{lemma:sym_gradient_vech_gradient}, we have (abusing notation when indexing \(s(x)\))
    \begin{equation*}
        s(x)_{ij} = (1 + \boole{i \neq j})S(x)_{ij}, \qquad \forall i,j \in [d]
    \end{equation*}
    Therefore, for any \(i,j,r,s \in [d]\) we have
    \begin{align*}
        \cI_{ij, rs} &=
        \expect{s(x)_{ij} s(x)_{rs}} 
        =
        (1 + \boole{i \neq j}) (1 + \boole{r \neq s}) \expect{S(x)_{ij} S(x)_{rs}}
        \\
        &= \frac{(1 + \boole{i \neq j}) (1 + \boole{r \neq s})}{4} \cdot \expect{(\Sigma^{-1} x x^\transp \Sigma^{-1} - \Sigma^{-1})_{ij} (\Sigma^{-1} x x^\transp \Sigma^{-1} - \Sigma^{-1})_{rs}}.
    \end{align*}
    Let \(e_i \in \{0,1\}^m\) (for dimension \(m > 0\) clear from context) denote the indicator vector given by \(e_i(j) \coloneq \boole{i = j}\) for any \(j \in [m]\). For ease of notation, define \(\Psi \coloneqq \Sigma^{-1}\) and \(\psi_i \coloneqq \Sigma^{-1} e_i\) for each \(i \in [d]\). Then we have
    \begin{align*}
        &\expect{(\Sigma^{-1} x x^\transp \Sigma^{-1} - \Sigma^{-1})_{ij} (\Sigma^{-1} x x^\transp \Sigma^{-1} - \Sigma^{-1})_{rs}}
        \\
        = &\expect{(\psi_i^{\transp} x x^{\T} \psi_j - \Psi_{ij})(\psi_r^{\transp} x x^{\T} \psi_s - \Psi_{rs}) }
        \\
        =&\expect{\psi_i^{\transp} x x^{\T} \psi_j \psi_r^{\transp} x x^{\T} \psi_s} - \Psi_{ij} \cdot \expect{\psi_r^\T x x^{\T} \psi_s}
        - \Psi_{rs} \cdot \expect{\psi_i^\T x x^{\T} \psi_j} + \Psi_{ij} \Psi_{rs}
        \\
        =& \psi_i^{\transp} \expect{ x x^{\T} \psi_j \psi_r^{\transp} x x^{\T}} \psi_s -  \Psi_{ij} \Psi_{rs},
    \end{align*}
    where in the last equation we used that \(\expectsmall{x x^{\T}} = \Sigma\) and the fact that \(\psi_p^\T \Sigma \psi_q = \Psi_{pq}\) for any \(p,q \in [d]\).
    By Lemma~\ref{lemma:fourth_moment} with \(B = \psi_j \psi_r^\T\) yields
    \begin{align*}
        \psi_i^{\transp} \expect{ x x^{\T} \psi_j \psi_r^{\transp} x x^{\T}} \psi_s
        &= \psi_i^\T \paren{\Sigma(\psi_j \psi_r^{\T} + \psi_r \psi_j^\T)\Sigma + \Sigma \Tr(\psi_j \psi_r^\T \Sigma)} \psi_s
        \\
        &= \psi_i^\T \paren{\Sigma(\psi_j \psi_r^{\T} + \psi_r \psi_j^\T)\Sigma + \Sigma \Psi_{rj}} \psi_s
        \\
        &=  \Psi_{ij} \Psi_{rs} + \Psi_{ir} \Psi_{js} + \Psi_{is} \Psi_{rj}. 
    \end{align*}
    Therefore, we conclude that 
    \begin{align*}
        &\expect{(\Sigma^{-1} x x^\transp \Sigma^{-1} - \Sigma^{-1})_{ij} (\Sigma^{-1} x x^\transp \Sigma^{-1} - \Sigma^{-1})_{rs}}
        \\ 
        &= \Psi_{ij} \Psi_{rs} + \Psi_{ir} \Psi_{js} + \Psi_{is} \Psi_{rj} - \Psi_{ij} \Psi_{rs}
        \\
        &= \Psi_{ir} \Psi_{js} + \Psi_{is} \Psi_{rj}
        \\
        &= (\Psi \otimes \Psi)_{ij, rs} + (\Psi \otimes \Psi)_{ij, sr}
        \\ 
        &= \frac{1}{2}\vec\paren*{e_i^{} e_j^\T + e_j e_i^{\T}}^\T (\Psi \otimes \Psi) \vec\paren*{e_r^{} e_s^\T + e_s e_r^{\T}}
    \end{align*}
    where in the last equation we used that \(\Psi\) is symmetric and, thus, \(\Psi_{ir} \Psi_{js} + \Psi_{is} \Psi_{rj} = \Psi_{jr} \Psi_{is} + \Psi_{js} \Psi_{ri}\). Finally, assuming without loss of generality that \(i \leq j\), the claim follows since
    \begin{equation*}
        \frac{(1 + \boole{i \neq j})}{2} \vec(e_i e_j^\T + e_j e_i^{\T})
        = \vec(e_i e_j^\T + e_j e_i^{\T} - e_i e_j \odot I) = D \vech(e_i e_j^{\T}),
    \end{equation*}
    where \(\odot\) the last equation follows from the definition of the duplication matrix and since \(i \leq j\) (see \citealp[Def.~3.2a]{MagnusN80a})
\end{proof}

\begin{lemma}
    \label{lemma:fisher_information}
    Let \(\Sigma \succ 0\). Then \(\cI(\Sigma) = \frac{1}{2} D^{\transp}  \Sigma^{-1} \otimes \Sigma^{-1} D \) where \(\otimes\) denotes the Kronecker product between matrices. In particular, we have \(\lambda_{\max}(\cI(\Sigma)) \leq  \lambda_{\min}(\Sigma)^{-2}\).    
\end{lemma}
\begin{proof}
    From Proposition~\ref{prop:fisher_info_formula} we have
    \begin{equation*}
        \cI(\Sigma) = \frac{1}{2}  D^{\transp}  \Sigma^{-1} \otimes \Sigma^{-1} D.
    \end{equation*}
    Moreover, one may note that \(\norm{D x}_2^2 \leq 2 \norm{x}_2^2\) for any \(x \in \bR^{d(d+1)/2}\). Therefore, 
    \begin{align*}
        \lambda_{\max}(\cI(\Sigma)) 
        &=
        \max\setst{\tfrac{1}{2}x^{T} D^{\transp}  \Sigma^{-1} \otimes \Sigma^{-1} D  x}{x \in \bR^{d(d+1)/2}, \norm{x}_2 \leq 1}
        \\
        &\leq \max\setst{z^{T}  \Sigma^{-1} \otimes \Sigma^{-1} z}{x \in \bR^{d(d+1)/2}, \norm{z}_2 \leq 1}
        = \lambda_{\max}(\Sigma^{-1} \otimes \Sigma^{-1})
    \end{align*}
    Finally, the result then follows by~\citet[Thm.~4.2.12]{HornJ91a} which shows that the set of all eigenvalues of \(\Sigma^{-1} \otimes \Sigma^{-1}\) is \(\setst{\lambda_i(\Sigma^{-1}) \cdot \lambda_j(\Sigma^{-1})}{i,j \in [d]}\).
\end{proof}

\section{Mathematical Background}

\subsection{Results from Probability and Statistics}

We will use the following tail bound for the $\chi^2$ distribution.

\begin{lemma}[{\citealp[Lemma~1]{LaurentM00}}]
    \label{lemma:chi_squared_concentration}
    Let \(Z \sim \chi^2(d)\) and \(x > 0\). Then
    \begin{equation*}
        \prob{Z - d \geq 2 \sqrt{dx} + 2 x} \leq e^{-x}.
    \end{equation*}
\end{lemma}

We now derive a corollary of this bound that will be more convenient for our purposes.

\begin{corollary}
\label{cor:chi_squared}
Let \(Z \sim \chi^2(d)\) and \(x > 0\). Then
    \begin{equation*}
        \prob{Z \geq \sqrt{8d^2 + 18x^2}} \leq e^{-x}.
    \end{equation*}
\end{corollary}
\begin{proof}
We have
\begin{align*}
0 &~\leq~ (2d-3x)^2
  ~=~ 4 d^2 - 12 d x + 9 x^2 \\
  &~=~ (8 d^2 + 18 x^2) - (4 d^2 + 12 dx + 9 x^2) \\
  &~=~ (8 d^2 + 18 x^2) - (2d + 3x)^2.
\end{align*}
Rearranging and using the AM-GM inequality, we have
\[
(8d^2 + 18x^2)^{1/2}
    ~\geq~ 2d + 3x ~=~ d + 2x + (d+x)
    ~\geq~ d+2x + 2 \sqrt{dx}.
\]
Thus, $\prob{ Z \geq (8d^2 + 18x^2)^{1/2} }
\leq \prob{ Z \geq d + 2 \sqrt{dx} + 2x } \leq e^{-x} $
by Lemma~\ref{lemma:chi_squared_concentration}.
\end{proof}


\subsection{Properties of Wishart matrices}

In this section we collect known results for Wishart matrices and derive corollaries for the normalized Wishart distribution as the one of \(\Sigma\) in~\eqref{eq:def_sigma}. First, for the lower bound in Section~\ref{sec:LB_haff} we shall need a few properties of the inverse Wishart distribution, which are collected in the following lemma.

\begin{lemma}[{\citealp[Theorem~3.2]{Haff79a}}]
    \label{lemma:properties_inv_wishart}
    Let \(\Sigma \sim \cW_d(D; V)\) for some non-singular \(V \in \PSDCone\) and with \(D > d + 3\). Then \(\expect{\Sigma^{-1}} = \frac{1}{D - d -1} V^{-1}\). Moreover, for every distinct \(i,j \in [d]\) we have
    \begin{align*}
        \var{\Sigma_{ii}^{-1}} &= 
        \frac{
            2 (V_{ii}^{-1})^2}{(D - d -1)^2(D - d - 3)} 
        \quad \text{and} \quad
        \\
        \var{\Sigma_{ij}^{-1}} &= 
            \frac{
                (D - d +1)(V_{ij}^{-1})^2 + (D - d -1)V_{ii}^{-1} V_{jj}^{-1}}{(D - d -1)^2(D - d - 3) (D-d)}.
    \end{align*}
\end{lemma}

Let us know collect a few facts about the normalized Wishart distribution, mainly regarding concentrations of its eigenvalues.

\begin{proposition}
    Let \(\Sigma \coloneqq \frac{1}{D} G G^{\transp}\) with the entries of \(G \in \bR^{d \times D}\) being i.i.d.\ standard Gaussians. Then \(\expect{\Sigma} = I\) and \(\expect{\norm{\Sigma - I}_F^2} = d^2/D\). In particular, if \(D = 2d\) then \(\expect{\norm{\Sigma - I}_F^2} = d/2\).
\end{proposition}
\begin{proof}
    Let \(g_i \sim \cN(0,I)\) be the \(i\)-th row of \(G\). Then \(\Sigma_{ij} = \frac{1}{D}\iprod{g_i}{g_j}\). Since \(g_i\) is independent of \(g_j\) for \(i \neq j\), one can see that \(\expect{\Sigma}_{ij} = I\). Moreover, 
    \begin{align*}
        \expect{\norm{\Sigma - I}_F^2}
        =  \sum_{i,j \in [d]} \expect{ \frac{1}{D} \paren*{\iprod{g_i}{g_j} - \ones\boole{i = j}}^2} 
        = \frac{1}{D^2} \sum_{i,j \in [d]} \expect{ \paren*{\iprod{g_i}{g_j} - D \ones\boole{i = j}}^2} 
    \end{align*}
    If \(i = j\), then \(\expect{ \paren*{\iprod{g_i}{g_j} - D \ones\boole{i = j}}^2} = D\) since it is exactly the variance of a \(\chi^2(D)\) distribution. For \(i \neq j\), we have 
    \begin{equation*}
        \expect{ \iprod{g_i}{g_j}^2 }
        = \sum_{k = 1}^D \underbrace{\expect{g_i(k)^2 g_j(k)^2}}_{= 1 \cdot 1 = 1} + \sum_{r,s \in [D], r \neq s}^d \underbrace{\expect{g_i(r) g_j(r) g_i(s) g_j(s) }}_{ = 0} = D.
    \end{equation*}
    Therefore, \(\expect{\norm{\Sigma - I}_F^2} = d^2 D/D^2 = d^2/D\) as desired.
\end{proof}

The following theorem gives us tail bounds on the singular values of a random Gaussian matrix, which will yield sub-exponential tails for \(\lambda_{\max}(\Sigma)\).

\begin{theorem}[{\citealt[Thm.~6.1]{Wainwright19}}]
    \label{thm:sing_value_gaussian}
    Let \(\Sigma \in \bR^{d \times d}\) be positive definite  and let \(G\) be a \(d \times D\) random matrix with i.i.d.\ columns each with distribution \(\cN(0, \Sigma)\).  Then, for all \(\delta > 0\),
    \begin{equation*}
        \prob{\frac{\sigma_{\max}(G^{\transp})}{\sqrt{D}} \geq \lambda_{\max}(\Sigma^{1/2})(1 + \delta) + \sqrt{\frac{\Tr(\Sigma^{1/2})}{D}}} \leq \exp\paren[\Big]{-\frac{D}{2} \cdot \delta^2},
    \end{equation*}
    where \(\sigma_{\max}(G^{\transp})\) is the maximum singular value of \(G^{\transp}\). Moreover, if \(D > d\), then 
    \begin{equation*}
        \prob{\frac{\sigma_{\min}(G^{\transp})}{\sqrt{D}} \leq \lambda_{\min}(\Sigma^{1/2})(1 - \delta) - \sqrt{\frac{\Tr(\Sigma^{1/2})}{D}}} \leq \exp\paren[\Big]{-\frac{D}{2} \cdot \delta^2},
    \end{equation*}
    where \(\sigma_{\min}(G^{\transp})\) is the minimum singular value of \(G^{\transp}\).
\end{theorem}

\begin{lemma}
    \label{lemma:tail_bound_lambdamax_sigma}
    Let \(G\) be a \(d \times D\) random matrix with \(D = 2d\) and i.i.d.\ entries each with distribution \(\cN(0,1)\).  Define the matrix \(W \coloneqq \frac{1}{D} G G^{\transp}\). Then, for any \(\delta > 0\) we have
    \begin{equation}
        \label{eq:tail_bound_lambda_max_sigma}
        \prob{\lambda_{\max}(\Sigma) \geq 6 +  2 \delta^2}  \leq \prob{\lambda_{\max}(\Sigma) \geq \paren*{1 + \frac{1}{\sqrt{2}} + \delta}^2} \leq \exp(-d \delta^2).
    \end{equation}
    In particular, we have
    \begin{equation*}
        \prob{\lambda_{\max}(\Sigma) \geq 10} \leq e^{-2d} \leq \frac{1}{\sqrt{d}} e^{-d}.
    \end{equation*}
\end{lemma}
\begin{proof}
    The first inequality in~\eqref{eq:tail_bound_lambda_max_sigma}
    holds since \((1 + 1/\sqrt{2} + \delta^2) \leq 2(1 + 1/\sqrt{2})^2 + 2 \delta^2 \leq 6 + 2 \delta^2\). The second inequality in~\eqref{eq:tail_bound_lambda_max_sigma} follows directly from Theorem~\ref{thm:sing_value_gaussian} by noticing that each column of \(G\) has distribution \(\cN(0, I)\). Thus, since \(I^{1/2} = I\) and \(\Tr(I) = d = D/2\), we have
\begin{align*}
    \prob{\lambda_{\max}(\Sigma) \geq \paren*{1 + \frac{1}{\sqrt{2}} + \delta}^2}
    &= \prob{\paren*{\frac{1}{\sqrt{D}} \cdot \sigma_{\max}(G^{\transp})}^2 \geq \paren*{\paren[\Big]{1 + \sqrt{\frac{\Tr(I^{1/2})}{D}}} + \delta}^2}.
\end{align*}
In particular, define \(\delta \coloneqq\sqrt{10} - 1 - 1/\sqrt{2}\). Then, using the last inequality and the fact that \(\delta^2 \leq 2\),
\begin{equation*}
    \prob{\lambda_{\max}(\Sigma)} = \prob{\lambda_{\max}(\Sigma) \geq \paren*{1 + \frac{1}{\sqrt{2}} + \delta}^2} \leq \exp(-d \delta^2) \leq e^{-2d},
\end{equation*}
as desired.
\end{proof}

For our purposes, the tails bounds from Theorem~\ref{thm:sing_value_gaussian} are too weak to usefully bound \(\lambda_{\min}(\Sigma)\). The reason for that is that the tail bound on \(\prob{\lambda_{\min}(\Sigma)< t} \) does not vanish as \(t\) goes to 0, although we know \(\lambda_{\min}(\Sigma) > 0\) almost surely. In other words, we want to provide tail bounds on \(1/\lambda_{\min}(\Sigma)\). Thus, to better control \(\lambda_{\min}(\Sigma)\) and the condition number \(\lambda_{\max}(\Sigma)/\lambda_{\min}(\Sigma)\) we will use results by~\citet{ChenD05a}.

\begin{lemma}[{\citealt[Lemma~4.1]{ChenD05a}}]
    \label{lemma:wishart_tail_bound}
    Let \(W\) be a \(d \times d\) Wishart matrix with \(D\) degrees of freedom. Then, for any \(t > 0\)
    \begin{equation*}
        \prob{\lambda_{\min}(W) \leq \frac{D}{t^2}}
        < \frac{1}{\Gamma(D - d + 2)} \paren*{\frac{D}{t}}^{D - d + 1} = \frac{1}{(D - d + 1)!} \paren*{\frac{D}{t}}^{D - d + 1}.
    \end{equation*}
\end{lemma}

\begin{lemma}
    \label{lemma:bound_expected_inv_lambda_min}
    Define \(\Sigma \coloneqq \frac{1}{D} G G^{\transp}\) where \(G\) is a \(d \times D\) random matrix with i.i.d.\ standard Gaussian entries with \(D = 2d\) and \(d \geq 10\). Then
    \begin{equation}
        \label{eq:tail_bound_lambda_min_sigma}
        \prob{\frac{1}{\lambda_{\min}(\Sigma)} \geq t}
         \leq \frac{(2e)^{d+1}}{\sqrt{2\pi d}}  \frac{1}{t^{d+1}}.
    \end{equation}
    In particular, \(\prob{\lambda_{\min}(\Sigma) < 0.09} \leq d^{-1/2} 2^{-d}\). Moreover, we have
    \begin{equation*}
        \expect{\frac{1}{\lambda_{\min}(\Sigma)}} \leq 2e + 1 \leq 6.5.
    \end{equation*}
\end{lemma}
\begin{proof}
    Since \(D \cdot \Sigma = G G^{\transp}\) follows a Wishart distribution with \(D\) degrees of freedom,   Lemma~\ref{lemma:wishart_tail_bound} yields for any \(t > 0\) the bound 
    \begin{equation*}
        \prob{\frac{1}{\lambda_{\min}(\Sigma)} \geq t}
        \leq \frac{1}{(D - d + 1)!} \paren*{\frac{D}{t}}^{D - d + 1} = \frac{1}{(d + 1)!} \paren*{\frac{2d}{t}}^{d + 1}.
    \end{equation*}
    Using a non-asymptotic estimate of Stirling's approximation, we have
    \begin{equation*}
        (d+1)! \geq \sqrt{2 \pi (d+1)} \paren*{\frac{d+1}{e}}^{d+1} \geq \sqrt{2 \pi d}\paren*{\frac{d}{e}}^{d+1}. 
    \end{equation*}
    Therefore,
    \begin{equation*}
        \prob{\frac{1}{\lambda_{\min}(\Sigma)} \geq t}
        \leq 
        \frac{1}{\sqrt{2\pi d}} 
        \paren*{\frac{2ed}{d+1}}^{d + 1} \frac{1}{t^{d+1}} \leq \frac{1}{\sqrt{2\pi d}} (2e)^{d+1} \frac{1}{t^{d+1}},
    \end{equation*}
    which proves~\eqref{eq:tail_bound_lambda_min_sigma}. In particular, we have
    \begin{equation*}
        \prob{\lambda_{\min}(\Sigma) \leq 0.09} \leq \prob{\lambda_{\min}(\Sigma) \leq \frac{1}{4e}} \leq \frac{1}{\sqrt{2\pi d}} \paren*{\frac{2e}{4e}}^{d+1} = \frac{1}{\sqrt{2\pi d}} -2^{-(d+1)}\leq \frac{1}{\sqrt{d}} 2^{-d}.
    \end{equation*}
    Finally, we can use~\eqref{eq:tail_bound_lambda_min_sigma} to upper bound \(\expect{1/\lambda_{\min}(\Sigma)}\) by
    \begin{equation*}
        \expect{\frac{1}{\lambda_{\min}(\Sigma)}}
        \leq \int_{0}^{\infty} \prob{\frac{1}{\lambda_{\min}(\Sigma)} \geq t} \diff t
        \leq 2e + \int_{2e}^{\infty} \prob{\frac{1}{\lambda_{\min}(\Sigma)} \geq t} \diff t.
    \end{equation*}
    For the last integral in the right-hand side, we have
    \begin{align*}
        \int_{2e}^{\infty} \prob{\frac{1}{\lambda_{\min}(\Sigma)} \geq t} \diff t
        &\leq \int_{2e}^{\infty}  \frac{1}{\sqrt{2\pi d}} \paren*{2e}^{d + 1} \frac{1}{t^{d+1}} \diff t
        =
        \int_{1}^{\infty}  \frac{1}{\sqrt{2\pi d}} 
        \frac{1}{y^{d+1}}  \cdot 2e \diff y
        =\frac{2e}{\sqrt{\pi d}} \frac{1}{d} \leq 1,
    \end{align*}
    where in the last inequality we used \(d \geq 10 \geq 2e\).
\end{proof}

\begin{lemma}[{\citealt[Theorem~4.5]{ChenD05a}}]
    \label{lemma:wishart_condition_number_tail}
    Let \(W\) be a \(d \times d\) Wishart matrix with \(D\) degrees of freedom. Then, there is a constant \(C \leq 6.414\) independent of \(d\) and \(D\) such that, for any \(t > 0\),
    \begin{equation*}
        \prob{\sqrt{\frac{\lambda_{\max}(W)}{\lambda_{\min}(W)}} > \frac{D}{D - d + 1} \cdot t}
        < \frac{1}{\sqrt{2 \pi}} \paren*{\frac{C}{t}}^{D - d + 1}.
    \end{equation*}
\end{lemma}

\begin{lemma}
    \label{lemma:tail_cond_number_sigma}
    Let \(\Sigma\) follow a normalized Wishart distribution as in~\eqref{eq:def_sigma}. Then
    \begin{equation*}
        \prob{\frac{\lambda_{\max}(\Sigma)}{\lambda_{\min}(\Sigma)} > t} < \frac{(13)^{d+1}}{\sqrt{2\pi}} \cdot \frac{1}{t^{(d+1)/2}}
    \end{equation*}
\end{lemma}
\begin{proof}
    Define \(\kappa(\Sigma) \coloneqq \lambda_{\max}(\Sigma)/\lambda_{\min}(\Sigma)\). Since \(D = 2d\), we have \(D/(D - d + 1) = 2d/(d+1) \geq 2\). Thus, for all \(t > 0\) it follows from Lemma~\ref{lemma:wishart_condition_number_tail} that
    \begin{align*}
        \prob{\kappa(\Sigma) > t}
        & = \prob{\sqrt{\kappa(\Sigma)} > \frac{2d}{d+1} \cdot \frac{d+1}{2d}\cdot \sqrt{t}}
        \\
        & \leq \prob{\sqrt{\kappa(\Sigma)} > \frac{2d}{d+1} \cdot \frac{1}{2} \cdot \sqrt{t}}
        < \frac{1}{\sqrt{2\pi}} \paren*{\frac{C}{\sqrt{t}/2}}^{d+1},
    \end{align*}
    where \(C \leq 6.414\) is as in Lemma~\ref{lemma:wishart_condition_number_tail}. The final bound follows by noting that \(2\cdot C \leq 13\).
\end{proof}

\acks{We thank Matthew Scott, Sharvaj Kubal, Yaniv Plan, and Naomi Graham for several helpful technical discussions.
We thank Abner Turkieltaub, Betty Shea, Curtis Fox, Chris Liaw and Frederik Kunstner for their feedback on early versions of the manuscript.}


\bibliography{bib.bib}

\end{document}